\documentclass[11pt]{article}

\usepackage[numbers]{natbib} 

\usepackage{titlesec}
\usepackage{mathtools}
\usepackage{bbm}
\usepackage{xcolor}
\usepackage{amsfonts,amsmath,amssymb,color,amsthm,boxedminipage,url,fullpage,xparse}
\definecolor{weborange}{rgb}{.8,.3,.3}
\definecolor{webblue}{rgb}{0,0,.8}
\definecolor{internallinkcolor}{rgb}{0,.5,0}
\definecolor{externallinkcolor}{rgb}{0,0,.5}

\usepackage{amsthm}
\usepackage{stackengine}
\usepackage{comment}
\usepackage{enumerate,paralist}
\usepackage[labelfont=bf]{caption}
\usepackage{aliascnt}

\ifdefined\asdvi
	\usepackage[dvips,
	hyperfootnotes=false,
	colorlinks=true,
	urlcolor=externallinkcolor,
	linkcolor=internallinkcolor,
	filecolor=externallinkcolor,
	citecolor=internallinkcolor,
	breaklinks=true,
	pdfstartview=FitH,
	pdfpagelayout=OneColumn]{hyperref} 
\else
	\usepackage[
	hyperfootnotes=false,
	colorlinks=true,
	urlcolor=externallinkcolor,
	linkcolor=internallinkcolor,
	filecolor=externallinkcolor,
	citecolor=internallinkcolor,
	breaklinks=true,
	pdfstartview=FitH,
	pdfpagelayout=OneColumn]{hyperref}
\fi

\usepackage{amsfonts,amsmath,amssymb,amsthm,boxedminipage,color,url,fullpage}
\usepackage[normalem]{ulem}
\usepackage{cleveref}
\usepackage{xspace}
\usepackage{xstring}
\usepackage{tikz}
\usetikzlibrary{arrows}
\usetikzlibrary{decorations.markings}
\usepackage{subcaption}
\usepackage{tabularx}
\usepackage{bbold}
\usepackage{dsfont}
\usepackage{titlesec}
\usepackage{enumitem}
\usepackage{nicefrac}

\providecommand{\remove}[1]{}
\newcommand{\Draft}[1]{\ifdefined\IsDraft\texttt{ #1} \fi}

\titleclass{\subsubsubsection}{straight}[\subsection]

\newcounter{subsubsubsection}[subsubsection]
\renewcommand\thesubsubsubsection{\thesubsubsection.\arabic{subsubsubsection}}

\titleformat{\subsubsubsection}
{\normalfont\normalsize\bfseries}{\thesubsubsubsection}{1em}{}
\titlespacing*{\subsubsubsection}
{0pt}{3.25ex plus 1ex minus .2ex}{1.5ex plus .2ex}

\makeatletter
\renewcommand\paragraph{\@startsection{paragraph}{5}{\z@}%
	{3.25ex \@plus1ex \@minus.2ex}%
	{-1em}%
	{\normalfont\normalsize\bfseries}}
\renewcommand\subparagraph{\@startsection{subparagraph}{6}{\parindent}%
	{3.25ex \@plus1ex \@minus .2ex}%
	{-1em}%
	{\normalfont\normalsize\bfseries}}
\def\toclevel@subsubsubsection{4}
\def\toclevel@paragraph{5}
\def\toclevel@paragraph{6}
\def\l@subsubsubsection{\@dottedtocline{4}{7em}{4em}}
\def\l@paragraph{\@dottedtocline{5}{10em}{5em}}
\def\l@subparagraph{\@dottedtocline{6}{14em}{6em}}
\makeatother

\newcommand{\sdotfill}{\textcolor[rgb]{0.8,0.8,0.8}{\dotfill}} 

\newenvironment{protocol}{\begin{proto}}{\end{proto}}

\newcommand{\aka} {also known as }

\newcommand{\ie}{i.e.,\xspace}
\newcommand{\eg}{e.g.,\xspace}

\newcommand{\wrt} {with respect to\xspace}
\newcommand{\wlg} {without loss of generality\xspace}

\newcommand{\abs}[1]{\left\lvert #1 \right\rvert}

\newcommand{\set}[1]{\ens{#1}}

\newcommand{\eqdef}{:=}

\newcommand{\N}{{\mathbb{N}}}

\newcommand{\zo}{\set{0,1}}

\makeatletter
\@ifpackageloaded{complexity}{}{%
\usepackage[classfont=roman,langfont=roman,funcfont=roman]{complexity}}
\makeatother

\renewcommand{\cref}{\Cref}
\usepackage{aliascnt}

\newtheorem{theorem}{Theorem}[section]

\newaliascnt{lemma}{theorem}
\newtheorem{lemma}[lemma]{Lemma}
\aliascntresetthe{lemma}
\crefname{lemma}{Lemma}{Lemmas}

\newaliascnt{claim}{theorem}
\newtheorem{claim}[claim]{Claim}
\aliascntresetthe{claim}
\crefname{claim}{Claim}{Claims}

\newaliascnt{corollary}{theorem}

\aliascntresetthe{corollary}
\crefname{corollary}{Corollary}{Corollaries}

\newaliascnt{construction}{theorem}

\aliascntresetthe{construction}
\crefname{construction}{Construction}{Constructions}

\newaliascnt{fact}{theorem}
\newtheorem{fact}[fact]{Fact}
\aliascntresetthe{fact}
\crefname{fact}{Fact}{Facts}

\newaliascnt{proposition}{theorem}
\newtheorem{proposition}[proposition]{Proposition}
\aliascntresetthe{proposition}
\crefname{proposition}{Proposition}{Propositions}

\newaliascnt{conjecture}{theorem}

\aliascntresetthe{conjecture}
\crefname{conjecture}{Conjecture}{Conjectures}

\newaliascnt{definition}{theorem}
\newtheorem{definition}[definition]{Definition}
\aliascntresetthe{definition}
\crefname{definition}{Definition}{Definitions}

\newaliascnt{notation}{theorem}
\newtheorem{notation}[notation]{Notation}
\aliascntresetthe{notation}
\crefname{notation}{Notation}{Notation}

\newaliascnt{assertion}{theorem}

\aliascntresetthe{assertion}
\crefname{assertion}{Assertion}{Assertion}

\newaliascnt{assumption}{theorem}

\aliascntresetthe{assumption}
\crefname{assumption}{Assumption}{Assumption}

\newaliascnt{remark}{theorem}
\newtheorem{remark}[remark]{Remark}
\aliascntresetthe{remark}
\crefname{remark}{Remark}{Remarks}

\newaliascnt{example}{theorem}
\ifdefined\excludeexample
\else
   \newtheorem{example}[example]{Example}
\fi

\aliascntresetthe{example}
\crefname{exmaple}{Example}{Examples}

\crefname{equation}{Equation}{Equations}

\newaliascnt{proto}{theorem}

\newtheorem{proto}[proto]{Protocol}

\aliascntresetthe{proto}
\crefname{proto}{protocol}{protocols}

\newaliascnt{expr}{theorem}
\newtheorem{expr}[expr]{Experiment}
\aliascntresetthe{expr}
\crefname{experiment}{experiment}{experiments}

\newcommand{\stepref}[1]{Step~\ref{#1}}

\def\FullBox{$\Box$}
\def\qed{\ifmmode\qquad\FullBox\else{\unskip\nobreak\hfil
\penalty50\hskip1em\null\nobreak\hfil\FullBox
\parfillskip=0pt\finalhyphendemerits=0\endgraf}\fi}

\def\qedsketch{\ifmmode\Box\else{\unskip\nobreak\hfil
\penalty50\hskip1em\null\nobreak\hfil$\Box$
\parfillskip=0pt\finalhyphendemerits=0\endgraf}\fi}

{\proof[#1]\list{}{\leftmargin=1cm\rightmargin=1cm}}
{\endproof%
  \vspace{10pt}
}

\makeatletter
\let\xx@thm\@thm
\AtBeginDocument{\let\@thm\xx@thm}
\makeatother

\newenvironment{algorithm}{\begin{algo}}{\vspace{-\topsep}\end{algo}}

\newaliascnt{algo}{theorem}
\newtheorem{algo}[algo]{Algorithm}
\aliascntresetthe{algo}
\crefname{algo}{algorithm}{algorithms}

\newcommand{\Ensuremath}[1]{\ensuremath{#1}\xspace}

\newcommand{\MathAlg}[1]{\mathsf{#1}}
\newcommand{\MathAlgX}[1]{\Ensuremath{\MathAlg{#1}}}

\newcommand{\tth}[1]{\Ensuremath{#1^{\rm th}}}

\newcommand{\ith}{\tth{i}}

\newcommand{\cI}{\mathcal{I}}

\newcommand{\cM}{{\cal{M}}}

\ifdefined\cS
\else
	\newcommand{\cS}{{\cal{S}}}
\fi

\newcommand{\cx}{{\cal{X}}}
\newcommand{\cY}{{\cal{Y}}}
\newcommand{\cX}{{\cal{X}}}
\newcommand{\cZ}{{\cal{Z}}}

\newcommand{\1}{\mathds{1}}

\DeclareMathOperator{\loglog}{loglog}

\newcommand{\Ac}{\MathAlgX{A}}
\newcommand{\Bc}{\MathAlgX{B}}
\newcommand{\Cc}{\MathAlgX{C}}
\newcommand{\Pc}{\MathAlgX{P}}

\newcommand{\tPi}{\widetilde{\Pi}}

\newcommand{\getsr}{\leftarrow}

\let\originalleft\left
\let\originalright\right
\renewcommand{\left}{\mathopen{}\mathclose\bgroup\originalleft}
\renewcommand{\right}{\aftergroup\egroup\originalright}

\providecommand{\remove}[1]{}

\newcommand{\nfrac}[2]{\nicefrac{#1}{#2}}

\crefname{ineq}{Inequality}{Inequalities}
\crefname{eq}{Equality}{Equalities}
\creflabelformat{ineq}{#2{\upshape(#1)}#3}
\creflabelformat{eq}{#2{\upshape(#1)}#3}

\newcommand{\Inqref}[1]{Inequality (\ref{#1})}
\newcommand{\Eqref}[1]{Equality (\ref{#1})}

\ifdefined\IsDraft
    \newcommand{\drafteddata}[1]{{#1}}
    
\else
    \newcommand{\drafteddata}[1]{}
    
\fi

\newcommand{\suchthat}{{: \;}}
\newcommand{\cond}{\,\vert\,}

\newcommand{\brackets}[1]{\left(#1\right)}

\newcommand{\bracketss}[2]{#1\brackets{#2}}
\newcommand{\braces}[1]{\left\{#1\right\}}
\newcommand{\ubraces}[1]{\{#1\}}
\renewcommand{\set}[1]{\braces{#1}}

\newcommand{\minnn}[1]{\bracketss{\min}{#1}}

\newcommand{\size}[1]{\left|#1\right|}
\newcommand{\ssize}[1]{\bigl|#1\bigl|}

\newcommand{\Ex}{\mathbb{E}}
\newcommand{\Var}{\mathrm{Var}}
\renewcommand{\Pr}{\mathbb{P}}

\newcommand{\pr}[1]{\Pr\bigl[#1\bigr]}

\newcommand{\expec}[2]{{\underset{#1}{\Ex}\bigl[#2\bigr]}}
\newcommand{\texpec}[2]{{\textstyle \Ex_{#1}\bigl[#2\bigr]}}
\newcommand{\expecc}[1]{\expec{}{#1}}
\newcommand{\varian}[1]{{\Var\bigl[#1\bigr]}}
\newcommand{\var}[1]{\Var[#1]}

\DeclareMathOperator{\Supp}{Supp}
\newcommand{\supp}[1]{\bracketss{\Supp}{#1}}

\newcommand{\KL}{\operatorname{KL}}
\newcommand{\kld}[2]{D_{\KL}\bigl({#1 \,\|\, #2} \xspace\bigr)}

\newcommand{\SD}{\mathsf{SD}}
\newcommand{\sdist}[2]{\bracketss{\SD}{#1, #2}}

\newcommand{\logg}[1]{\bracketss{\log}{#1}}

\newcommand{\ind}[1]{\mathbbm{1}_{#1}}
\newcommand{\restr}[2]{#1 \big\vert_{#2}}

\newcommand{\reals}{\mathbb{R}}

\makeatletter
\newrobustcmd*{\bigop@dot}[2]{%
   \setbox0=\hbox{$\m@th#1#2$}%
   \vbox{%
     \lineskiplimit=\maxdimen
     \lineskip=-0.7\dimexpr\ht0+\dp0\relax
     \ialign{%
       \hfil##\hfil\cr
       $\m@th\cdot$\cr
       \box0\cr
     }%
   }%
}
\providerobustcmd*{\bigcupdot}{%
   \mathop{%
     \mathpalette\bigop@dot\bigcup
   }%
}
\makeatother

\newcommand{\bigOmega}[1]{\bracketss{\Omega}{#1}}

\newcommand{\bigO}[1]{\bracketss{O}{#1}}

\newcommand{\eps}{\varepsilon}

\newcommand{\Tableofcontents}{
	\thispagestyle{empty}
	\pagenumbering{gobble}
	\clearpage
	\tableofcontents
	\thispagestyle{empty}
	\clearpage
	\pagenumbering{arabic}
}

\ifdefined\prettyformatflag
  \newcommand{\prettyformat}[1]{#1}
\else
  \newcommand{\prettyformat}[1]{}
\fi

\ifdefined\prettyformatflag  

\let\tinput\input 
\renewcommand{\input}[1]{
	\tinput{#1} 
}
\fi

\newcommand{\numparties}{n}
\newcommand{\lenproto}{\ell}

\newcommand{\condMsg}{Q}
\newcommand{\condMsgg}[1]{\condMsg_{#1}}

\DeclareMathOperator{\Msg}{Msg}
\newcommand{\Msgg}[1]{\Msg_{#1}}
\newcommand{\Msggg}[2]{\Msg^{#2}_{#1}}

\newcommand{\tMsg}{\widetilde{\Msg}}
\newcommand{\tMsgg}[1]{\tMsg_{#1}}

\newcommand{\CorruptedMsg}{\Msg^{\Ac}}
\newcommand{\CorruptedMsgg}[1]{\CorruptedMsg_{#1}}

\newcommand{\HonestMsg}{\Msg^{\Pi}}
\newcommand{\HonestMsgg}[1]{\HonestMsg_{#1}}

\DeclareMathOperator{\msg}{msg}
\newcommand{\msgg}[1]{\msg_{#1}}

\newcommand{\hatmsg}{\widehat{\msg}}

\newcommand{\aprot}[2]{{#1}_{#2}}
\newcommand{\Party}{\Pc}
\newcommand{\partition}{\Party}

\newcommand{\party}{\mathsf{NxtParty}}

\newcommand{\Length}[1]{\bracketss{\mathrm{NumMsgs}}{#1}}
\newcommand{\BadParty}{\MathAlgX{NonRbst}}
\newcommand{\Idx}{\mathrm{SentBy}}
\newcommand{\Speaker}{\mathrm{Speaker}}

\newcommand{\diff}{\mathrm{jump}}

\newcommand{\protinc}[1]{\diff^{#1}}
\newcommand{\protincc}[2]{\bracketss{\protinc{#1}}{#2}}

\DeclareMathOperator{\Biased}{Biased}
\newcommand{\biased}[3]{\Biased_{#1}^{#2}\bigl(#3\bigr)}
\newcommand{\cbiased}[4]{\biased{#1}{#2}{#3 \cond #4}}

\newcommand{\GoodJumps}{\mathrm{RobustJumps}}
\newcommand{\SmallJumps}{\mathrm{SmallJumps}}
\newcommand{\SmallParties}{\mathrm{SmallParties}}
\newcommand{\SmallContribParties}{\mathrm{SmallContribParties}}
\newcommand{\ContribParties}{\mathrm{ContribParties}}
\newcommand{\LargeJumps}{\mathrm{LargeJumps}}
\newcommand{\LargeParties}{\mathrm{LargeParties}}
\newcommand{\CorruptedParties}{\mathrm{CorruptedParties}}

\title{
	A Tight Lower Bound on \\
	Adaptively Secure Full-Information Coin Flip \\
	\Draft{\\{\small \sc Working Draft: Please Do Not Distribute}}
}

\author{
	Iftach Haitner \thanks{Stellar Development Foundation and The Blavatnik School of Computer Science, Tel Aviv University. Email: \texttt{iftachh@gmail.com}.}~\thanks{Research supported by ERC starting grant 638121, and Israel Science Foundation grants 666/19 and 836/23. }
	\and Yonatan Karidi-Heller\footnotemark[2]~\thanks{The Blavatnik School of Computer Science, Tel Aviv University. Email: \texttt{ykaridi@gmail.com.}}
}

\begin{document}
\sloppy
\maketitle

\begin{abstract}
In a distributed {coin-flipping} protocol, \citeauthor{Blum83} [ACM Transactions on Computer Systems '83], the parties try to output a common (close to) uniform bit, even when some adversarially chosen parties try to bias the common output. In an {adaptively secure full-information coin flip}, \citeauthor{BenOrL85} [FOCS '85], the parties communicate over a broadcast channel, and a {computationally unbounded} adversary can choose which parties to corrupt \emph{along} the protocol execution. \citeauthor{BenOrL85} proved that the $\numparties$-party majority protocol is resilient to $O(\sqrt{\numparties})$ corruptions (ignoring poly-logarithmic factors), and conjectured this is a tight upper bound for any $\numparties$-party protocol (of any round complexity). Their conjecture was proved to be correct for \textit{single-turn} (each party sends a single message) \textit{single-bit} (a message is one bit) protocols \citeauthor*{lichtensteinLS89} [Combinatorica '89], \textit{symmetric} protocols \citeauthor*{GoldwasserKP15} [ICALP '15], and recently for (arbitrary message length) single-turn protocols \citeauthor*{TaumanKR18} [DISC '18]. Yet, the question of many-turn protocols was left entirely open.

In this work, we close the above gap, proving that \emph{no} $\numparties$-party protocol (of any round complexity) is resilient to $\omega(\sqrt{\numparties})$ (adaptive) corruptions.
\end{abstract}

\allowdisplaybreaks

\Tableofcontents

\section{Introduction}
In a distributed (\aka collective) {coin-flipping} protocol, \citet{Blum83},
the parties try to output a common (close to) uniform bit, even when some adversarially chosen parties try to bias the output. 
Coin-flipping protocols are fundamental primitives in cryptography and distributed computation, 
allowing distrustful parties to agree on a common random string (\eg public randomness) to be used in their joint computation. 
More generally, almost any random process/protocol/algorithm encapsulates (some form of) a coin-flipping protocol. 
Consider a random process whose Boolean output is far from being fixed (\eg has noticeable variance). 
Such a process can be thought of as a coin-flipping protocol: the common coin is the output, and the process's randomness corresponds to the parties' messages. 
Thus, lower bounds on coin-flipping protocols induce limitations on the stability of random processes (see \cref{sec:intro:Poisoning} for a concrete example).

The focus of this work is \textit{full-information} coin-flipping protocols, \citet{BenOrL85}. 
In this variant, the parties communicate solely over a single broadcast channel, and the Byzantine adversary\footnote{Once it corrupts a party, it completely controls it and can send arbitrary messages on its behalf.} is assumed to be computationally \emph{unbounded}. 
Two types of such adversaries are considered: 
A \textit{static} adversary that chooses the parties it corrupts \emph{before} the execution begins, and an \textit{adaptive} adversary that can choose the parties it wishes to corrupt \emph{during} the protocol execution (\ie as a function of the messages seen so far). 
For static adversaries, full-information coin flips are well understood, and almost matching upper (protocols) and lower (attackers) bounds are known, see \cref{sec:intro:RelatedWork}.
Much less is understood about adaptive adversaries, which are the focus of this work, and significant gaps exist between the upper and lower bounds. 
\citet*{BenOrL85} proved that the $\numparties$-party majority protocol is resilient to $O(\sqrt{\numparties})$ corruptions (ignoring poly-logarithmic factors in $\numparties$), and conjectured that this is a tight upper bound for any $\numparties$-party protocol (\ie of any round complexity). 
The works of \citet*{lichtensteinLS89, GoldwasserKP15} made progress towards proving the conjecture for \textit{single-turn} (each party sends a single message) protocols, a case that was eventually proved by \citet*{TaumanKR18}. 
Yet, the question of many-turn protocols was left entirely open.

\subsection{Our Results} \label{sec:intro:Results}
We solve this intriguing question, showing that the output of any $\numparties$-party protocol can be \emph{fully biased} by an \emph{adaptive} adversary corrupting $O(\sqrt{\numparties})$ parties (ignoring poly-logarithmic factors).

\begin{theorem} [Biasing full-information coin-flipping protocols, informal] \label{thm:biasing_coinflips:Inf}
	For any $\numparties$-party full-information coin-flipping protocol, there exists $b\in\zo$ and an (unbounded) adversary that, by adaptively corrupting $O(\sqrt{\numparties})$ of the parties, forces the outcome of the protocol to $b$, except with probability $o(1)$.  
\end{theorem}

The above lower bound matches (up to poly-logarithmic factors) the upper bound achieved by the $\numparties$-party majority protocol \cite{BenOrL85}. 
The bound extends to biased protocols, \ie protocols with expected outcome (in an all-honest execution) different from $\nfrac12$.
We also remark that the one-side restriction (only possible to bias the protocol outcome to some $b\in\zo$) is inherent, as there exists, for instance, an $\numparties$-party (single-turn) protocol that is resilient to $\Theta(\numparties)$ corruptions trying to bias its outcome towards one.\footnote{
	{\label{fn:1} Consider the $\numparties$-party single-turn protocol in which each party broadcasts a ${(\nfrac 1\numparties,1- \nfrac 1\numparties)}$-biased bit (\ie equals zero with probability $\nfrac 1\numparties$) and the protocol output is set to the AND of these bits. 
	It is clear that the protocol expected outcome is $(1-\nfrac 1\numparties)^\numparties \approx 1/e$ (can be made $\nfrac 12$ by slightly changing the distribution), and that even $\numparties/2$ adaptive corruptions cannot change the protocol outcome to a value larger than $(1-\nfrac 1\numparties)^{\numparties/2} \approx \sqrt{1/e}$.}
}

Although the one-sidedness is inherent in the adaptive case, we can overcome this by allowing the adversary to perform \emph{strongly adaptive} corruptions, \ie the adversary can decide whether to corrupt a party \emph{after} seeing the message it is about to send.

\begin{theorem} [Biasing full-information coin-flipping protocols via strongly adaptive attacks, informal]
	For any $\numparties$-party full-information coin-flipping protocol and {\sf any non-insignificant output} $b\in\zo$, there exists an (unbounded) adversary that by performing at most $O(\sqrt{\numparties})$ strongly-adaptive corruptions, forces the outcome of the protocol to $b$, except with probability $o(1)$.
\end{theorem}

\subsection{Related Work} \label{sec:intro:RelatedWork}

\subsubsection{Full-Information Coin Flip}
We recall the main known results for $\numparties$-party full-information coin-flipping protocols. 
\paragraph{Adaptive adversaries.}
In the following, we ignore poly-logarithmic factors in $\numparties$.
\begin{description}
	\item[Upper bounds (protocols).] \citet*{BenOrL85} proved that the majority protocol is resilient to $O(\sqrt{\numparties})$ corruptions.
	
	\item[Lower bounds (attacks).] \citet*{lichtensteinLS89} proved that no \textit{single-bit} (messages are one bit) single-turn protocol is resilient to $\bigOmega{\sqrt{n}}$ adaptive corruptions (hence, majority is optimal for such protocols). \citet{Dodis2001} proved that it is impossible to create a coin-flipping protocol resilient to $\bigOmega{\sqrt{n}}$ adaptive corruptions by \emph{sequentially repeating} another coin-flipping protocol, and then applying a deterministic function to the outcomes. \citet*{GoldwasserKP15} proved that that no \emph{symmetric} single-turn (many-bit) protocol is resilient to $\bigOmega{\sqrt{n}}$ adaptive corruptions. Their result extends to \emph{strongly adaptive} attacks (the attacker can decide to corrupt a party \emph{after} seeing the message it is about to send) on single-turn protocols. \citet*{TaumanKR18} fully answered the single-turn case by proving that no single-turn protocol is resilient to $\bigOmega{\sqrt{n}}$ adaptive corruptions. Lastly, \citet*{EOMSM19} presented an \emph{efficient} strongly adaptive attack on protocols of certain properties (\eg public coins).
\end{description}

\paragraph{Static adversaries.}
The case of static adversaries is well studied and understood. 
\begin{description}
	\item[Upper bounds (protocols).] \citet{BenOrL85} presented a protocol that tolerates $\bigO{\numparties^{0.63}}$ corrupted parties (an improvement on the $O(\sqrt{\numparties})$ corrupted parties it takes to bias the majority protocol). \citet{AjtaiL93} presented a protocol that tolerates $\bigO{\nfrac{\numparties}{\log^2 \numparties}}$ corruptions. \citet{Saks89} presented a protocol that tolerates $\bigO{\nfrac{\numparties}{\log \numparties}}$ corruptions. The protocol of \cite{Saks89} was later improved by \citet{AlonN1993} to tolerate a constant fraction of corrupted parties. Shortly afterwards, \citet{boppana2000perfect} presented an optimal protocol resilient to $(\nfrac12 - \delta)\numparties$ corruptions for any $\delta>0$.
	
	\item[Lower bounds (attacks).] \citet*{KKL88} proved that no single-bit single-turn protocol can tolerate $\bigOmega{\nfrac{\numparties}{\log \numparties}}$ corruptions. \citet*{RSZ02} proved that a protocol tolerating $\bigOmega{\numparties}$ corruptions is either many-bit or has $\Omega(\nfrac12 - o(1))\cdot \log^* (\numparties)$ rounds.
\end{description}

The reader is referred to \citet{Dodis06} for a more elaborated, somewhat outdated survey on full-information coin flip and friends.

\subsubsection{Data-Poisoning Attacks} \label{sec:intro:Poisoning}
Consider a learning algorithm that tries to learn a hypothesis from a training set of samples from different sources. The random process corresponding to this learning task can be naturally viewed as a coin-flipping protocol. Moreover, as first noticed by \citet{MahloujifarM17}, an attacker on the latter coin flip induces a so-called \textit {data-poisoning attack}: increasing the probability of a desired property (\ie poisoning the training data) by tampering with a small number of sources. For this application, however, the attacker would better be able to force a predetermined output (rather than forcing some output, as our attack achieves). Hence, the attack we apply on the coin-flipping protocol should be \emph{bidirectional} (have the ability to (almost-) fully determine the coin, rather than biasing it to some arbitrary value). While this goal is unachievable in some models (see \cref{fn:1}), it is achievable in some important ones (\eg \cite{EtesamiMM20}).

Materializing their observation, \citet{MahloujifarM17} translated a two-directional \emph{static} attack into a static data-poisoning attack on learning algorithms. Their attack was further improved in \cite{MahloujifarMM2018,MahloujifarDM19a}. \citet{MahloujifarM2019} translated the two-directional \emph{adaptive} attack of \cite{TaumanKR18} on single-bit, single-turn coin-flipping protocols into an adaptive data-poisoning attack. Finally, \citet{EtesamiMM20} facilitated their strongly adaptive attack on single-turn coin-flipping protocols (see \cref{sec:intro:RelatedWork}) to obtain a strongly-adaptive data-poisoning attack.

The adaptive adversaries above cast attacks on \emph{single-turn} coin-flipping protocols into data-poisoning attacks that tamper some samples. With the tools we present (see \cref{thm:biasing_coinflips:Inf}), it is now possible to discuss data-poisoning attacks that scale with the amount \emph{sources}, rather than the amount of samples (which we may have orders of magnitude more of).

\subsection*{Open Questions}
In this work, we show that the expected outcome of \emph{any} $\numparties$-party full-information coin-flipping protocol can be biased to either $o(1)$ or to $1-o(1)$, using $O(\sqrt{\numparties})$ corruptions. The above $o(1)$, however, stands for $1/\loglog(\numparties)$, and it remains an intriguing question whether it can be pushed to $2^{-\polylog(\numparties)}$ as can be achieved, for instance, when attacking the $\numparties$-party majority protocol. Such attacks are known for \emph{uniform} single-bit single-turn protocols (a secondary result of \cite{TaumanKR18}) and for \emph{strongly adaptive} attackers against single-turn protocols \cite{EOMSM19}.

\subsection*{Paper Organization}
A rather elaborate description of our attack on coin-flipping protocols is given in \cref{sec:Technique}. Basic notations, definitions, and facts are given in \cref{sec:Preliminaries}. We also present there some useful manipulations of coin-flipping protocols. In \cref{sec:AttackingRobustProtocols}, we show how to attack protocols of certain structure, that we refer to as \textit{robust}, and in \cref{sec:arbitrary_coinflip} we extend this attack to arbitrary protocols.
Finally, in \cref{sec:StronglyAdaptive} we present a \emph{bidirectional} strongly adaptive attack (a bidirectional attack is impossible in the standard adaptive model).
\section{Our Technique}\label{sec:Technique}
This section gives a rather elaborate description and analysis of our adaptive attack on full-information coin-flipping protocols. Let $\Pi$ be an $\numparties$-party, $\lenproto$-round, full-information coin-flipping protocol. We prove that one can either bias the expected outcome of $\Pi$ to less than $\eps \eqdef 1/\loglog(\numparties)$, or to more than $1-\eps$.

Similarly to previous adaptive attacks on full-information coin-flipping protocols, our attack makes use of the ``jumps'' in the protocol's expected outcome; 
assume \wlg (see \cref{sec:prelim:fullcoin} for justification) that in each round only a \emph{single} party sends a message, and let $\Msg = (\Msgg{1},\ldots, \Msgg{\lenproto})$ denote the protocol transcript (\ie parties' messages) in a random all-honest execution of $\Pi$. 
For $\msg\in \Supp(\Msg)$, let $\Pi(\msg)$ denote the final outcome of the execution described by $\msg$. For $\msgg{\le i}\in \Supp(\Msgg{\le i} \eqdef (\Msgg{1},\ldots,\Msgg{i}))$ let $\Pi(\msgg{\le i}) \eqdef \expecc{\Pi(\Msg) \mid \Msgg{\le i} = \msgg{\le i})}$ be the expected outcome given a partial transcript, 
and let
\[
	\diff(\msgg{\le i}) \eqdef \Pi(\msgg{\le i}) - \Pi(\msgg{<i})
\]
be the ``jump'' in the expected outcome induced by the \ith message. 
Accordingly, we refer to $\diff(\Msgg{\le i})$ as the \ith jump in the protocol execution. 
Our attack manipulates these jumps in a very different manner than what previous attacks did. 
First, the decision whether to corrupt a given message is based on the (conditional) \emph{variance} of the jumps ($L_2$ norm), a more subtle measure than the \emph{maximal} possible change ($L_\infty$ norm) considered by previous attacks. 
Second, even when the attacker decides that the next message is useful for biasing the protocol's outcome, it only \emph{gently} alters the message: 
It corrupts the party about to send the message only with a certain probability, and even when corrupting, only moderately changes the message distribution. Being gentle allows the attack to bypass the main obstacles in attacking many-turn protocols. 
The gentleness also makes analyzing the attack's performance easier; 
the transcript of the gently attacked execution is not ``too different'' from the all-honest (un-attacked) execution of the protocol. 
Consequently, the analysis requires only a good understanding of the all-honest execution, and not of the typically very complicated execution the attack induces.\footnote{
	Gentle attacks, in the general sense that the attacker does not try to maximize the effect of the attack in each round but instead keeps the attacked execution similar to the all-honest one, were found helpful in many other settings. A partial list includes attacking different types of coin-flipping protocols \cite{Maji10,HaitnerOmri14,BermanHT18,BHMO18}, and proving parallel repetition of computationally sound proofs \cite{HaastadPWP10,HaitnerPR13,BermanHT20}.
} 
Details below.

We start by describing an attack on \emph{robust} protocols: for some $b\in \zo$, the protocol has \emph{no} $\nfrac{1}{\sqrt{\numparties}}$ jumps towards $b$,\footnote{
	An almost accurate example of a (bidirectional) robust protocol is the single-bit, single-turn, $\numparties$-party \textit{majority} protocol: in each round, a single party broadcasts an unbiased coin, and the protocol's final output is set to the majority of the coins. 
	It is well known that the absolute value of most jumps is (typically) order of $\nfrac{1}{\sqrt{\numparties}}$.
} and for concreteness, we assume $b=0$. That is, for every $i$ (w.p.\ one)
\begin{align}\label{eq:intro:Technique:1}
\Pi(\Msgg{\le i}) \ge \Pi(\Msgg{< i}) - \nfrac{\eps^2}{\sqrt{\numparties}}
\end{align}

We first consider single-turn robust protocols and then generalize to many-turn (robust) protocols. The extension to arbitrary (non-robust) protocols is described in \cref{sec:intro:Technique:NonRobust}.

\subsection{Attacking Robust Single-Turn Coin Flip}\label{sec:intro:Technique:SingleRobust}

Our attack gently biases the expected outcome of the protocol by carefully manipulating the message distributions of the corrupt parties. This manipulation (to be applied to the party's next message distribution) is defined as follows.

\begin{definition}[$\Biased$ distribution]
	For a distribution $P$,  $\alpha \ge 0$ and mapping $f \colon \supp{P} \mapsto [-\nfrac{1}{\alpha}, \infty)$ with $\expecc{f(P)} = 0$, let $\biased{\alpha}{f}{P}$ be the distribution defined by 
	\[
		\pr{\biased{\alpha}{f}{P} = e} \eqdef \pr{P = e} \cdot \brackets{1 + \alpha \cdot f(e)}\enspace
	\]
\end{definition}

That is, the distribution $P$ is ``nudged'' towards larger values of $f$, \ie increasing the probability of positive elements (causing $\expecc{\bracketss{f}{\biased{\alpha}{f}{P}}} \ge \expecc{f(P)}$), and the larger $\alpha$ is the larger the bias is.
Equipped with this definition, our attacker on a ($\numparties$-party, single-turn, $\numparties$-round, robust) coin-flipping protocol $\Pi$ is defined in as follows:

\begin{algorithm}[Single-turn attacker, informal]\label{alg:intro:Technique:Single}
	\smallskip
	For $i=1$ to $\numparties$, do the following \emph{before} the \ith message is sent:
	\begin{enumerate}
		\item Let $\msgg{< i}$ be the previously sent messages. Let $\condMsgg{i} \eqdef \Msgg{i}|_{\Msgg{<i} = \msgg{< i}}$, $\diff_i \eqdef \diff(\msgg{< i}, \cdot)$, let $v_i \eqdef \var{\diff_i(\condMsgg{i})}$.

		\item If $v_i \ge \nfrac{\eps^4}{\numparties}$, corrupt the \ith party with probability $\nfrac{1}{\eps^4} \cdot \sqrt{v_i}$. If corrupted, instruct it to send the next message according to $\biased{\nfrac{1}{\sqrt{v_i}}}{\diff_i}{\condMsgg{i}}$. \label{alg:intro:Technique:Single:corruption}
	\end{enumerate}
\end{algorithm}

That is, a party is corrupted with probability proportional to the (conditional) standard deviation it induces on the expected outcome of $\Pi$ (\ie messages inducing larger variance on the protocol outcome are more likely to be corrupted). 
If corrupted, the message distribution is modified so that the change it induces on the expected outcome of $\Pi$ is biased towards one, where the bias is proportional to the inverse of the standard deviation (\ie messages with smaller variance are leveraged more aggressively, to ``compensate'' for the fact they have small variance).

\begin{example}[Attacking single-turn majority]
	If $\Pi$ is the single-turn, $\numparties$-party, single-bit, majority protocol, then (typically) each $v_i$ is of (absolute) order $\nfrac{1}{\numparties}$. 
	Thus, in expectation, the above attack corrupts $\nfrac{1}{\eps^3} \cdot \sqrt{\numparties}$ parties. 
	If corrupted, the party's bit message is set to $1$ with probability $\approx \nfrac12 \cdot(1 + \sqrt{\numparties}\cdot \nfrac{1}{\sqrt{\numparties}}) = 1$.
\end{example}

In the following, we argue that the attacker indeed biases the expected outcome of $\Pi$ to $1 - \eps$ and that the expected number of corruptions is $\widetilde{O}(\sqrt{\numparties})$. 
Therefore, a Markov bound yields the existence of the required attacker. 
We prove the success of our attack by showing that the attacked protocol has too little ``liveliness'' to resist the attacker's bias. Consequently, the outcome is (with high probability) the value the attacker biases towards. 
Our notion of liveliness is the \emph{conditional variance} of some underlying distribution induced by the attack.

Let $\CorruptedMsg = (\CorruptedMsgg{1},\ldots, \CorruptedMsgg{\numparties})$ be the message distribution induced by the above attack.
For $i\in [\numparties]$, let $\condMsgg{i}$ be the value of $\condMsgg{i}$ in the attacked execution, determined by $\CorruptedMsgg{i-1}$, and sample $Y_i \leftarrow \diff(\CorruptedMsgg{<i}, \condMsgg{i})$. 
$\condMsgg{i}$ correspond to an honest message distribution, therefore, $\expecc{Y_i \mid \CorruptedMsgg{<i}} = 0$. Hence, the sequence $Y = (Y_1,\ldots,Y_n)$ is a \textit{martingale difference sequence} \wrt $(\CorruptedMsgg{i}, Y_i)_{i=1}^\numparties$.\footnote{
	The $Y_i$'s are sampled such that they are independent of each other even when conditioned on $\CorruptedMsg$.}  
This martingale can be seen as an honest execution based on a corrupted history. 

We show that $Y_1, \dots, Y_n$ have little ``liveliness'': their overall impact on the outcome is small. It follows protocol's outcome is determined solely by the adversary's manipulations. Those manipulations push the outcome towards $1$, so we conclude that the protocol outcome must be $1$. 
The core of our analysis lies in the following lemma, where we argue about the liveliness of $Y_1, \dots, Y_n$ (defined as the sum of conditional variances of $Y_1, \dots, Y_n$).

\begin{lemma}\label{lem:intro:Technique:Single}
	$\expecc{\sum_{i=1}^\numparties \var{Y_i \mid \CorruptedMsgg{<i}}}\le \eps^3$.
\end{lemma}

The proof of \cref{lem:intro:Technique:Single} is sketched below, but we first use it to analyze the attack's quality. Think of $\var{\sum_{i = 1}^\numparties Y_i}$ as the protocol ``liveliness'' described before. 
$Y$ is a martingale difference sequence \wrt $\CorruptedMsgg{i}$, \ie $\expecc{Y_i \mid \CorruptedMsgg{< i}} =0$, and so it is easy to verify that
\begin{align*}
	\varian{\sum_{i = 1}^\numparties Y_i} &= \sum_{i = 1}^\numparties \varian{Y_i} = \sum_{i = 1}^\numparties (\expecc{\varian{Y_i \cond \CorruptedMsgg{< i}}} + \var{\expecc{Y_i \cond \CorruptedMsgg{< i}}}) \nonumber \\
	&= \sum_{i = 1}^\numparties \expecc{\varian{Y_i \cond \CorruptedMsgg{< i}}} = \expecc{\sum_{i = 1}^\numparties \varian{Y_i \cond \CorruptedMsgg{< i}}} \le \epsilon^3.
\end{align*}
The first equality follows by \cref{fact:martingales_pythagoras} (martingale increments are orthogonal), the second by the law of total variance (\cref{fact:total_variance}), and the last inequality by \cref{lem:intro:Technique:Single}.
Thus, by Chebyshev's inequality
\begin{align} \label{technique:bound_yi}
	\pr{\ssize{{\sum_{i=1}^\numparties Y_i}} \ge \eps} \le \eps
\end{align}
Consider the \emph{sub}-martingale $S = (S_0,\ldots,S_n)$ \wrt $\braces{\CorruptedMsgg{i}}_{i=1}^\numparties$ defined by $S_i \eqdef \Pi(\CorruptedMsgg{\le i})$, \ie the expected honest outcome given $\CorruptedMsgg{\le i}$. 
By definition, $S_0 = \expecc{\Pi (\Msgg{{}})} = \nfrac12$ and $S_n = \Pi (\CorruptedMsg) \in \zo$.
In addition, the attack only \emph{increases} the conditional expectation $\expecc{S_{i + 1} - S_i \cond \CorruptedMsgg{\le i}}$ and originally the protocol jumps have (conditional) expectation of zero,
hence, it holds that $\expecc{S_{i + 1} - S_i \cond \CorruptedMsgg{\le i}} \ge 0$. So indeed, $S$ constitutes a \emph{sub}-martingale sequence. 
Since ($S_{i + 1} - S_i$) is a ``biased towards one'' variant of $Y_i$, there exists a (rather) straightforward coupling between $S$ and $Y$ for which
\begin{align*}
	S_i- S_{i-1} \ge Y_i.
\end{align*}
By definition $S_0= 1/2$, and thus $\pr{S_n \le 0} = \pr{\sum_{i = 1}^\numparties S_i - S_{i - 1} \le -\nfrac12}$. By the properties of the aforementioned coupling, $\pr{\sum_{i = 1}^\numparties (S_i - S_{i - 1}) \le -\nfrac12} \le \pr{\sum_{i=1}^\numparties Y_i  \le -\nfrac12}$, and by \cref{technique:bound_yi}, $\pr{\sum_{i = 1}^\numparties (S_i - S_{i - 1}) \le -\nfrac12} \le \eps$. 
Finally, $S_n \in \zo$ therefore, $\pr{S_n = 1} = 1 - \pr{S_n \le 0} \ge 1 - \eps$. Namely, the output of the attacked protocol is $1$ with a probability of at least $ 1-\eps$.

We conclude the attack analysis by bounding the expected number of corruptions. By construction, the probability the attacker corrupts the \ith party is at most
\begin{align*}
	\nfrac1{\eps^4} \cdot \sqrt{\var{Y_i \mid \CorruptedMsgg{<i}}} &= \nfrac1{\eps^4} \cdot \nfrac{\var{Y_i \mid \CorruptedMsgg{<i}}}{\sqrt{\var{Y_i \mid \CorruptedMsgg{<i}}}} \nonumber \\
	&\le \nfrac{1}{\eps^4} \cdot (\var{Y_i \mid \CorruptedMsgg{<i}} \cdot \sqrt{\nfrac{\numparties}{\eps^4}}) = \nfrac{\sqrt \numparties}{\eps^6} \cdot \var{Y_i \mid \CorruptedMsgg{<i}}.\nonumber
\end{align*}

Overall, the total amount of corruptions is at most
\begin{align*}
	\expecc{\sum_{i=1}^\numparties \nfrac{\sqrt \numparties}{\eps^6} \cdot \var{Y_i \mid \CorruptedMsgg{<i}}} = \nfrac{\sqrt \numparties}{\eps^6} \cdot \expecc{\sum_{i=1}^\numparties \var{Y_i \mid \CorruptedMsgg{<i}}}.
\end{align*}

Consequently, by \cref{lem:intro:Technique:Single}, the expected number of corruptions is at most $\nfrac{\sqrt{\numparties}}{\eps^6} \cdot \eps^3 = \nfrac{\sqrt{\numparties}}{\eps^3} = \widetilde{O}(\sqrt{\numparties})$.

\paragraph{Proving \cref{lem:intro:Technique:Single}.}
By definition of $\Biased$, for any $p \in [0, 1]$ it holds that
\begin{align} \label{eq:lem:intro:Technique:Single:2}
	\brackets{p \cdot \biased{\alpha}{f}{P} + \brackets{1 - p} \cdot P} \equiv \biased{p\cdot \alpha}{f}{P}
\end{align}

Also note that for any distribution $P$, constant $\alpha \ge 0$ and function $f\colon \supp{P} \mapsto [-\nfrac{1}{\alpha}, \infty)$ with $\expecc{f(P)} =0$, it holds that 
\begin{align} \label{eq:lem:intro:Technique:Single:1}
	 & \expecc{f\bigl(\biased{\alpha}{f}{P}\bigr)} = \sum_{x \in \supp{P}} f(x) \cdot \pr{\biased{\alpha}{f}{P} = x} \\
	 & = \sum_{x \in \supp{P}} f(x) \cdot \pr{P = x} \cdot (1 + \alpha \cdot f(P))\nonumber \\
	 & = \expecc{f(P) \cdot (1 + \alpha \cdot f(P))} = \expecc{f} + \alpha \cdot \expecc{f^2(P)} = 0 + \alpha \cdot \varian{f(P)}.\nonumber
\end{align}

\def\maybecorrupt{C}
Let $V_i$ be the value of the variable $v_i$ in the execution of the attack (determined by $\CorruptedMsgg{< i}$), and let $\maybecorrupt_i$ be the event $\braces{V_i \ge \nfrac{\eps^4}{\numparties}}$, {\ie whether the adversary \emph{even considers} to corrupt the current party (see \stepref{alg:intro:Technique:Single:corruption} of the attacker)}. 
Using the above notation, proving \cref{lem:intro:Technique:Single}  translates to proving that
\begin{align}
\expecc{\sum_{i = 1}^n V_i} \le \eps^3
\end{align}
Let $\msgg{< i} \in \Supp(\CorruptedMsgg{< i})$ be a partial transcript. 
Denote $\nu_i = \var{Y_i \cond \CorruptedMsgg{<i} = \msgg{< i}}$ (\ie matching $v_i$ or $V_i$). 
For  transcripts satisfying $\nu_i \ge \nfrac{\eps^4}{\numparties}$ (\ie the event $\maybecorrupt_i$ holds), applying \cref{eq:lem:intro:Technique:Single:1,eq:lem:intro:Technique:Single:2} \wrt 
\begin{align*}
P \eqdef  \restr{\Msgg{i}}{\Msgg{< i} = \msgg{< i}},\; p \eqdef \nfrac{1}{\eps^4} \cdot \sqrt{\nu_i}, \; \alpha \eqdef \nfrac{1}{\sqrt{\nu_i}} \text{~~and~~}  \diff_i \eqdef \diff(\msgg{< i}, \cdot),
\end{align*}
  yields
\begin{align}
	\expecc{S_i - S_{i-1} \mid \CorruptedMsgg{< i} = \msgg{< i}} & = \expecc{\diff_i\bigl(\biased{\nfrac{1}{\eps^4}}{\diff_i}{\restr{\Msgg{i}}{\Msgg{< i} = \msgg{< i}}}\bigr)} \label{eq:intro:Technique:Single:biased_single} \\
	& = \nfrac{1}{\eps^4} \cdot \varian{\diff(\Msgg{\le i}) \cond \Msgg{< i} = \msgg{< i}} \nonumber \\
	& = \nfrac{1}{\eps^4} \cdot \varian{Y_i \cond \CorruptedMsgg{< i} = \msgg{< i}}. \nonumber
\end{align}

Hence,
\begin{align}
	\expecc{S_i - S_{i-1} \mid \CorruptedMsgg{< i}} = \ind{\maybecorrupt_i} \cdot \expecc{S_i - S_{i-1} \mid \CorruptedMsgg{< i}}
	= \ind{\maybecorrupt_i} \cdot \nfrac{1}{\eps^4} \cdot V_i
\end{align}
The first equality holds by construction (if $\maybecorrupt_i$ doesn't hold, the party is not corrupted and thus expectation is zero). The second equality by \cref{eq:intro:Technique:Single:biased_single}.
Therefore,
\begin{align}
	\expecc{S_\numparties- S_0} = \expecc{\sum_{i=1}^\numparties (S_i - S_{i - 1})} = \expecc{\sum_{i=1}^\numparties \expecc{S_i - S_{i-1} \mid \CorruptedMsgg{< i}}} \le \nfrac{1}{\eps^4} \cdot \expecc{\sum_{i=1}^\numparties \ind{\maybecorrupt_i} \cdot V_i}
\end{align}
The second equality follows by the law of total expectation (\cref{fact:total_expectation}). Since $S_0$ and $S_\numparties$ take values in $[0,1]$, it follows that 
\begin{align}
	\expecc{\sum_{i=1}^\numparties \ind{\maybecorrupt_i} \cdot V_i} \le \eps^4 \label{technique:single:expec_maybecorrupted}
\end{align} 

In addition, by definition of $\maybecorrupt_i$, for any $i$ it holds that
\begin{align}
	\ind{\neg \maybecorrupt_i} \cdot V_i \le \nfrac{\eps^4}{\numparties} \label{technique:single:expec_noncorrupted}
\end{align}

Combining \cref{technique:single:expec_maybecorrupted,technique:single:expec_noncorrupted}, we deduce that $\expecc{\sum_{i=1}^\numparties V_i} \le \eps^4 + \nfrac{\eps^4}{\numparties} \cdot \numparties < \eps^3$.

\subsection{Attacking Robust Many-Turn Coin Flip}\label{sec:intro:Technique:ManyRobust}
When attacking many-turn coin-flipping protocols, one encounters two additional problems:
\begin{description}
	\item[Identify influential parties.] In many-turn protocols, each message might have \emph{very little influence on the protocol's outcome}. 
	So, it is unclear for an online attacker to decide which parties to corrupt; 
	\eg a party that sends insignificant messages at the beginning of the protocol might turn out to be influential in the future, or on the flip side a party that had a significant influence on the protocol will not necessarily have significant influence in the rest of the execution. 
	
	\item[Preserve corrupted parties' influence.] Even if we have successfully identified potentially influential parties in the protocol, we must not alter their behavior in a way that makes it obvious they are corrupted. 
	If the corrupted parties' messages differ vastly from the honest execution, it might no longer be the case that those parties stand a chance at significantly influencing the protocol's outcome. 
\end{description}
Let us exemplify the above obstacles using the following two examples, respectively. 

\begin{example} [Shrinking majority]
	Consider a ``shrinking'' $\numparties$-party $\numparties^2$-round majority protocol: a majority protocol consisting of $\numparties$ \textit{super}-rounds in which every player sends a single bit and in addition, a random (determined by the rounds' coins) party is cast out (meaning from this super-round on, its bits are ignored). In such a case, the attacker must decide whether to corrupt a party without being certain that it will ``survive'' for many rounds.
\end{example}

\begin{example}[Punishment mechanism]
	Consider the $\numparties$-party $\numparties^2$-round majority protocol, \ie each party sends $\numparties$ bits, that is equipped with the following ``punishment'' mechanism: once a party's coins are ``too suspicious'', say contain a $1$-run of length $\log^2 \numparties$, its coins are ignored from this point on. So, once corrupting a party, the attack should not attempt to bluntly maximize the effect of the messages it sends. 
\end{example}

\noindent We tackle both problems (respectively) following two general ideas.

\begin{description}
	\item[Corrupt parties at random based on their past influence.]
	Our attacker decides at random whether to corrupt a party, based on their past influence. 
	That is, the corruption process can be viewed as a lottery: a party starts with a single ticket (\ie chance to become corrupted), and every time it contributes a certain amount of influence on the protocol it gets another ticket (\ie another chance to become corrupted).\footnote{
		In the actual attack, we formalize this approach by partitioning the parties into \emph{pseudo-parties} with bounded influence and corrupt each at random \emph{independently} of each other.
	} While the above approach does not identify an influential party before it starts affecting the outcome, it does so before the party significantly affects the outcome.

	\item[Gently modify corrupted parties.]
	Once deciding to corrupt a party, our attacker only \textit{subtly} alters its messages, like in the single-turn case where we use the $\Biased$ transformation to subtly alter the messages of the corrupted parties; such a gentle attack assures that we do not ``drift'' too much from a ``typical'' execution, and in particular, influential parties remain influential even in the attacked protocol. 
	
	\item[Highly influential messages.]
	In addition to that above, we bias highly influential messages exactly like the single-turn case, \ie the probability we corrupt the matching party is only related to the influence of the message---without taking into consideration the number of parties.
\end{description}

\subsubsection{The Attack}
The above intuition takes form as the following attacker (against a $\numparties$-party, $\lenproto$-round, robust protocol $\Pi$). 

\begin{algorithm}[Many-turn protocols attacker, informal]\label{alg:intro:Many}~
	\smallskip
	
	For $i=1$ to $\lenproto$, do the following \emph{before} the \ith message is sent:
	\begin{enumerate}
		\item Let $\msgg{< i}$ be the previously sent messages, let $\condMsgg{i} \eqdef \Msgg{\le i}|_{\Msgg{<i} = \msgg{< i}}$, let $\diff_i \eqdef \diff(\msgg{< i}, \cdot)$, let $v_i \eqdef \var{\diff_i(\condMsgg{i})}$. 
		
		\item If $v_i \ge \nfrac{\eps^4}{\numparties}$, corrupt the party sending the \ith message with probability $\nfrac{1}{\eps^4} \cdot \sqrt{v_i}$. If corrupted, instruct it to send its next message according to $\biased{\nfrac1{\sqrt{v_i}}}{\diff_i}{\condMsgg{i}}$. (\ie like in the singe-turn case.)
		
		Else, if the \ith message is the first message to be sent by the party, or the overall contribution of the messages it sent since the last time it was considered for corruption exceeded $\nfrac{\eps^4}{\numparties}$,\footnote{
			In the main body of the paper, we transform (in the attacker's head) arbitrary protocols into \textit{normal} protocols in which the overall influence of a party's small-jump messages is limited. Thus, each party needs to be tested for corruption only once.
		}
		corrupt this party with probability $\nfrac{1}{\eps^4} \cdot \nfrac{1}{\sqrt{\numparties}}$ (if it was decided not to corrupt the party, we consider this party honest from now on).

		If corrupted (at the last decision point), instruct it to send its next message according to $\biased{\sqrt{\numparties}}{\diff_i}{\condMsgg{i}}$.
	\end{enumerate}
\end{algorithm}
That is, a {large-jump party} is treated like in the single-turn case. In contrast, a small-jumps party is corrupted with (fixed) probability proportional to $\nfrac{1}{\sqrt{\numparties}}$ (when corrupted, \emph{all} messages of the party are modified).

\begin{example} [Attacking many-turn majority]
	Consider the $\numparties$-party, $\numparties^2$-round, single-bit majority protocol (in which each party sends $\numparties$ bits). 
	Typically, the change induced by any given message is of order $\nfrac{1}{\numparties}$. Consequently, each $v_i$ is of order $\nfrac{1}{\numparties}^2$, and each party will be independently corrupted with probability $\nfrac{1}{\eps^3} \cdot \nfrac{1}{\sqrt{\numparties}}$. Thus, in expectation, the above attack corrupts $\nfrac{1}{\eps^3} \cdot \sqrt{\numparties}$ parties. 
	If corrupt, each of the single-bit messages the party sends is $1$ with probability $\approx \nfrac12 \cdot(1 + \sqrt{\numparties}\cdot \nfrac{1}{\numparties}) = \nfrac12 + \nfrac1{2\sqrt{\numparties}}$
\end{example}

\paragraph{Analysis.}

The analysis of the above attack is similar to the single-turn case. Let $\CorruptedMsg = (\CorruptedMsgg{1},\ldots, \CorruptedMsgg{\lenproto})$, $S = (S_0,\ldots,S_\lenproto)$ and $Y = (Y_1,\ldots,Y_\lenproto)$ be as in the single-turn case. Similarly to the single-turn case, the core of the proof lies in the following lemma.

\begin{lemma}\label{lem:intro:Technique:Many}
	$\expecc{\sum_{i=1}^\lenproto \var{Y_i \mid \CorruptedMsgg{<i}}} = O(\eps^4)$.
\end{lemma}
The challenge in proving \cref{lem:intro:Technique:Many} is that, unlike the single-turn proof, it might be that the following does not hold:
\begin{align*}
	\expecc{S_i - S_{i-1} \mid \CorruptedMsgg{< i}=\msgg{<i}} \ge \nfrac{1}{\eps^4} \cdot \var{Y_i\mid \CorruptedMsgg{< i}=\msgg{<i}}.
\end{align*}
Indeed, let $V_i$ be the value of the variables $v_i$ in the execution of the attack described by $\CorruptedMsg$. 
Assume that conditioned on $\CorruptedMsgg{< i}= \msgg{<i}$, it holds that $V_i < \nfrac{\eps^4}{\numparties}$ and that a party \Party is about to send the \ith message. 
Unlike the single-turn case, the conditional probability that \Party is corrupted is no longer guaranteed to match the (non-conditional) probability that \Party is corrupted: 
the previous messages sent by \Party in $\msgg{< i}$ might \emph{leak} whether \Party is corrupted or not. 
If the latter happens, then (by the same argument we used for proving the lemma in the single-turn case) it might be that $\expecc{S_i - S_{i-1} \mid \CorruptedMsgg{< i}=\msgg{<i}} < \nfrac{1}{\eps^4} \cdot \var{Y_i\mid \CorruptedMsgg{< i}=\msgg{<i}}$. 
Fortunately, since we only slightly modify each small-jump message and assume no party's small messages are too influential (which holds for normal protocols), 
a KL-divergence argument yields that on average for such messages it holds that $\expecc{S_i - S_{i-1} \mid \CorruptedMsgg{i-1}=\msgg{<i}} = \Omega(\nfrac{1}{\eps^4} \cdot \varian{Y_i\mid \CorruptedMsgg{i-1}=\msgg{<i}})$, 
which suffices for the proof of the lemma to go through.

\subsection{Attacking Non-Robust Coin Flip}\label{sec:intro:Technique:NonRobust}
The high-level idea of attacking a non-robust protocol (a protocol that has large jumps in both directions) is to attempt biasing the protocol towards zero in such a way that if this bias fails, the protocol will be robust for $b=0$ (no large jumps downward)---prime for applying the attack on robust protocols.
More formally, assume that with probability at least $\nfrac 1{\log \numparties}$, the protocol $\Pi$ has a large negative jump, \ie less than $-\nfrac 1{\sqrt{\numparties}}$, and consider the following ``one-shot'' attacker on $\Pi$:\footnote{Assuming the next-message function of $\Pi$ is efficiently samplable, \eg $\Pi$ is public-coin, the following attacker is the only reason for the inefficiency of our attack.}

\begin{algorithm}[Negative jumps attacker] For $i=1$ to $\numparties$, do the following \emph{before} the \ith message is sent:

	\begin{enumerate}
		\item Let $\msgg{< i}$ be the previously sent messages.

		\item If there exists $m_i^{-} \in \Supp(\Msgg{i}|_{\Msgg{<i} = \msgg{< i}})$ such that $\Pi( \msgg{< i},m_i^{-}) < \Pi(\msgg{< i}) -\nfrac{1}{\sqrt{\numparties}}$, and no party was corrupted yet, corrupt the party sending the \ith message and instruct it to send $m_i^{-}$.
	\end{enumerate}
\end{algorithm}

It is clear that above adversary biases the outcome of $\Pi$ toward zero by at least $\nfrac{1}{\sqrt{\numparties}\cdot \log(\numparties)}$. Let $\Pi_1$ be the \emph{protocol} induced by the above (deterministic) attack: all parties emulate the attacker in their head, and when it decides to (deterministically) corrupt a party, the corrupted party follows its (deterministic) instructions. If the protocol $\Pi_1$ has a large negative jump with probability larger than $\nfrac1{\log \numparties}$, apply the above attack on $\Pi_1$ resulting in the protocol $\Pi_2$, and so on. Let $t\le \sqrt{\numparties} \cdot \log(\numparties)$ denote the number of times we applied the attack in this manner. If the expected outcome of $\Pi_t$ is at most $\eps$, then we are done: the implied $t$-adaptive adversary makes $\Pi$ output $0$ with probability $1-\eps$. Otherwise, $\Pi_t$ has the following property:
\begin{align}
	\pr{\exists j\in [\numparties] \colon \Pi_t(\tMsgg{\le j}) < \Pi_t(\tMsgg{< j}) -\nfrac{1}{\sqrt{\numparties}}} \le 1/\log(\numparties)
\end{align}
letting $\tMsg$ be the messages of a random execution of $\Pi_t$. If the above happens, we apply the attack on robust protocols (\cref{alg:intro:Many}) on $\Pi_t$, instructing the adversary to halt if it encounters a large negative jump.

With careful analysis (and slight modifications to the attacker), one can show that, due to the \emph{gentleness} of the attack, the property of encountering large negative jumps only with negligible probability is preserved.
Hence, the attack carries as if there are no large negative jumps, meaning we have successfully biased the expected output of $\Pi_t$ to $1-O(\eps)$. Composing the attack that transforms $\Pi$ into $\Pi_t$ with the attack on robust protocols
(on $\Pi_t$) yields the required attack on $\Pi$.\footnote{In \cref{sec:intro:Technique:SingleRobust} we proved the quality of our attack for protocols with $\expecc{\Pi} = \nfrac12$. Still, the proof can be easily adapted to the case of $\expecc{\Pi} \ge \nfrac{1}{\polylog(\numparties)}$.}

\section{Preliminaries} \label{sec:Preliminaries}

\subsection{Notations}
We use calligraphic letters to denote sets, uppercase for random variables, and lowercase for values and functions. All logarithms considered here are base $2$. For $n\in \N$, let $[n] \eqdef \set{1,\ldots,n}$. Given a Boolean statement $S$ (\eg $X \geq 5$), let $\1_{S}$ be the indicator function that outputs $1$ if $S$ is a true statement and $0$ otherwise. For a distribution $X$, let $x\getsr X$ denote that $x$ was sampled according to $X$.

\subsection{Distributions and Random Variables} \label{sec:prelim:dist}

The support of a distribution $P$ over a discrete set $\cx$, denoted $\Supp(P)$, is defined by $\Supp(P) \eqdef \set{x\in \cx: P(x)>0}$. For random variables $X$ and $ Y$, let the random variable $\supp{X \cond Y}$ denote the conditional support of $X$ given $Y$. In addition, we define the random variables $\expecc{X \cond Y}$ and $\varian{X \cond Y}$ as (deterministic) functions of $Y$, by $\bracketss{\expecc{X \cond Y}}{y} \eqdef \expecc{X \cond Y = y}$ and $\bracketss{\varian{X \cond Y}}{y} \eqdef \varian{X \cond Y = y}$, respectively.

The {\sf statistical distance} (\aka {\sf variation distance}) of two distributions $P$ and $Q$ over a discrete domain $\cx$ is defined by $\sdist{P}{Q} \eqdef \max_{\cS\subseteq \cx} \size{P(\cS)-Q(\cS)} = \frac{1}{2} \sum_{x \in \cS}\size{P(x)-Q(x)}$. Statistical distance enjoys a data-processing inequality.
\begin{fact} [Data-processing inequality for statistical distance] \label{fact:dataprocess}
	For distributions $P$ and $Q$ and function $f$ over a discrete domain $\cx$, it holds that $\sdist{f(P)}{f(Q)} \le \sdist{P}{Q}$.
\end{fact}

The {\sf KL-divergence} (\aka {\sf ~Kullback-Leibler divergence} and {\sf relative
		entropy}) between two distributions $P$ and $Q$ over a discrete domain $\cX$,
is defined by
\begin{align*}
	\kld{P}{Q} \eqdef \sum_{x\in\cX}P(x)\log\frac{P(x)}{Q(x)} = \Ex_{x \getsr P}\log\frac{P(x)}{Q(x)},
\end{align*}
where $0\cdot\log\frac00 = 0$, and $\kld{P}{Q}\eqdef \infty$ if there exists $x\in\cX$ such that
$P(x)>0$ but $Q(x)=0$. KL-divergence is convex in the following sense:
\begin{fact} [Convexity of KL-divergence] \label{fact:kl_convex}
	For finite distributions $P_1, P_2, Q_1, Q_2$, and $\lambda \in [0, 1]$ it holds that $\kld{\lambda \cdot P_1 + (1 - \lambda) \cdot P_2}{\lambda Q_1 + (1 - \lambda) \cdot Q_2} \le \lambda \cdot \kld{P_1}{Q_1} + (1 - \lambda) \cdot \kld{P_2}{Q_2}$.
\end{fact}

In addition, KL-divergence enjoys a chain rule.
\begin{fact} [KL-divergence chain rule] \label{fact:kl_chain_rule}
	For distributions $P(X, Y)$ and $Q(X, Y)$ for a pair of random variables $X$ and $Y$, it holds that
	$\kld{P(X, Y)}{Q(X, Y)} = \kld{P(X)}{Q(X)} + \expec{x \leftarrow P(X)}{\kld{P(Y \cond X = x)}{Q(Y \cond X = x)}}$.
\end{fact}

The following fact (see \citet{FHT03}) relates small KL-divergence to small statistical distance:

\begin{fact} [Pinsker bound] \label{fact:pinsker_bound}
	For discrete distributions $P$ and $Q$ it holds that $\sdist{P}{Q} \le \sqrt{\frac{1}{2} \cdot \kld{P}{Q}}$.
\end{fact}

\subsection{Martingales}
Martingales play an important role in our analysis.
\begin{definition} [Martingales] \label{def:Martingale}
	A sequence of random variables $M= (M_1, \ldots, M_n)$ is a {\sf martingale} \wrt a sequence of random variables $X_1, \ldots, X_{n}$, if $\expec{}{M_{k+1} \cond X_{\le k}} = M_k$ and $M_k$ is determined by $X_{\le k}$ (for every $k\in [n]$). The sequence $M$ is a martingale if it is a martingale \wrt itself. The {\sf increments} (\aka differences) sequence of $M$ are the random variables $\set{M_{k+1} - M_k}_{k=1}^{n-1}$.
\end{definition}

In particular, we will be interested in the so-called Doob martingales.
\begin{definition} [Doob martingales]\label{def:DoobMartingale}
	The {\sf Doob martingale} of the random variables $X = (X_1, \ldots, X_n)$ induced by the function $f\colon \supp{X} \mapsto \reals$, is the sequence $M_1, \ldots, M_n$ defined by $M_k \eqdef \expecc{\bracketss{f}{X_1, \ldots, X_n} \cond X_{\le k}}$.
\end{definition}

The proof of the following known facts is immediate.
\begin{fact} [Martingale increments are orthogonal] \label{fact:martingales_pythagoras}
	Let $X_1, \ldots, X_n$ be a sequence of random variables. If there exist random variables $Z_1, \ldots, Z_n$ such that $\expecc{X_k \cond Z_{< k}} = 0$ and $X_k$ is determined by $Z_{\le k}$ (\ie $\sum X_i$ is a martingale \wrt $Z_k$), then $\var{\sum_{i=1}^n X_i} = \sum_{i=1}^n \var{X_i}$.
\end{fact}

\begin{fact} [Law of total expectation] \label{fact:total_expectation}
	For two random variables $Y, X$ it holds that $\expecc{Y} = \expecc{\expecc{Y \cond X}}$.
\end{fact}

\begin{fact} [Law of total variance] \label{fact:total_variance}
	For two random variables $Y, X$ it holds that $\var{Y} = \expecc{\var{Y \cond X}} + \var{\expecc{Y \cond X}}$.
\end{fact}

\paragraph{Sub-martingales.}
We also use the related notion of sub-martingales.
\begin{definition} [Sub-martingales] \label{def:submartingale}
	A sequence of random variables $S= (S_1, \ldots, S_n)$ is a {\sf sub-martingale} \wrt a sequence of random variables $X_1, \ldots, X_{n}$, if $\expec{}{S_{k+1} \cond X_{\le k}} \ge S_k$ and $S_k$ is determined by $X_{\le k}$ (for every $k\in [n]$). The sequence $S$ is a {\sf sub-martingale} if it is a sub-martingale \wrt itself.
\end{definition}

In particular, we make use of the following known inequality. 
\begin{fact} [Doob's maximal inequality] \label{fact:doob_maximal_inequality}
	Let $S_1, \ldots, S_n$ be a non-negative sub-martingale, then for any $c > 0$ it holds that $\pr{\sup_{k} {S_k} \ge c} \le \nfrac{\expecc{S_n}}{c}$. 
\end{fact}

\subsection{Full-Information Coin Flip} \label{sec:prelim:fullcoin}
We start with the formal definition of full-information coin-flipping protocols.
\begin{definition} [Full-information coin-flipping protocols] \label{def:full_info_coinflips}
	A protocol $\Pi$ is a {\sf full-information coin-flipping protocol} if it is {\sf stateless} (\ie the parties keep no private state between the different communication rounds),\footnote{Since we consider attackers of \emph{unbounded} computational power, this assumption is \wlg: given a stateful protocol we can apply our attack on its stateless variant in which each party, samples its state \emph{conditioned on the current public transcript} before it acts. It is easy to see that an attack on the stateless variant implies an attack of the same quality on the original (stateful) protocol.} single turn (\ie each turn consists of a {\sf single} party broadcasting a string), and the parties' common output is a deterministic Boolean function of the transcript.
\end{definition}

\begin{remark}[Many messages per communication round]
	Our attack readily applies to the model in which many parties might broadcast a message in a single round, as long as the adversary controls the message arrival order in this round (as assumed in \citet{TaumanKR18}). The setting in which many messages per round are allowed, and the adversary has no control over the arrival order, is equivalent (at least under a natural formulation of this model) to the static adversary cases, in which we know that $\Theta(\numparties/\log \numparties)$ corruptions ($\numparties$ being the number of parties) are required.
\end{remark}

\begin{notation}
	We associate the the following notation with an $\numparties$-party, $\lenproto$-party, full-information coin-flipping protocol $\Pi$:

	\begin{itemize}
		\item Let $\HonestMsg = (\HonestMsgg{1},\ldots,\HonestMsgg{\lenproto})$ denote a random transcript (\ie parties' messages) of $\Pi$.

		\item For partial transcript $\msgg{\le i} \in \supp{\HonestMsgg{\le i}}$, let $\bracketss{\Pi}{\msgg{\le i}} \eqdef \expecc{\bracketss{\Pi}{\HonestMsg} \cond \HonestMsgg{\le i} = \msgg{\le i}}.$
		(I.e., the expected outcome of $\Pi$ given $\msgg{\le i}$.) We let $\expec{}{\Pi} \eqdef \Pi()$, and refer to this quantity as the {\sf expected outcome of $\Pi$}.

		\item For $\msgg{< i}\in \Supp(\Msgg{< i}^\Pi)$, let $\party(\msgg{< i}) \in [\numparties]$ be the identity of the party to send the \ith message, as determined by $\msgg{< i}$.

		\item For a party $\Party \in [\numparties]$ and transcript $\msg\in \Supp(\HonestMsg)$, let $\Idx_\Party(\msg)\eqdef \set{i\in [\lenproto] \colon \party(\msgg{<i}) = \Party}$.
		
		\item For $\msgg{\le j}\in \Supp(\Msgg{< i}^\Pi)$ and $i < j$ let $\Speaker_i(\msgg{\le j}) = \party(\msgg{< i})$.

		\item For $\msgg{\le i}\in \Supp(\Msgg{\le i}^\Pi)$, let $\protincc{\Pi}{\msgg{\le i}} \eqdef \bracketss{\Pi}{\msgg{\le i}} - \bracketss{\Pi}{\msgg{< i}}$.

		(I.e., $\protincc{\Pi}{\msgg{\le i}} $ is the increment in expectation caused by the \tth{i} message.)
	\end{itemize}
\end{notation}

\subsubsection{Adaptive Adversaries}

\begin{definition} [Adaptive adversary]
	A {\sf $t$-adaptive adversary} for a full-information coin-flipping protocol in an {\sf unbounded} algorithm that can take the following actions during the protocol execution.

	\begin{enumerate}
		\item {\sf Before} each communication round, it can decide to add the next to speak party to the corrupted party list, as long as the size of this list does not exceed $t$.

		\item In a communication round where a corrupted party is speaking, the adversary has {\sf full control} over the message it sends but is bound to send a valid message (\ie in the protocol message space support).
	\end{enumerate}
\end{definition}

We make use of the following definitions and properties for such adversaries.

\paragraph{The attacked protocol.}
\begin{definition}[The attacked protocol]
	Given a full-information coin-flipping protocol $\Pi$ and a deterministic (adaptive) adversary \Ac attacking it, let $\Pi_\Ac$ be the full-information coin-flipping protocol induced by this attack: the parties act according to $\Pi$ while emulating \Ac. Once a party realizes it is corrupted, it acts according to the instruction of (the emulated) \Ac. For non-deterministic \Ac, let $\Pi_\Ac$ be the distribution over protocols induced by the randomness of \Ac.
\end{definition}

\paragraph{Derandomization.}
\begin{proposition} [Attacker derandomization] \label{prop:attacker_determinization}
	For an adversary \Ac acting on a full-information coin-flipping protocol $\Pi$ there exist {\sf deterministic} adversaries $\Ac^+$ and $\Ac^-$ such that $\expecc{\aprot{\Pi}{\Ac^+}} \ge \expecc{\aprot{\Pi}{\Ac}}$ and $\expecc{\aprot{\Pi}{\Ac^-}} \le \expecc{\aprot{\Pi}{\Ac}}$.
\end{proposition}
\begin{proof}
	By simple expectation arguments over the randomness of \Ac.
\end{proof}

\paragraph{Composition of adaptive adversaries.}
\begin{definition} [Adaptive adversary composition] \label{def:adaptive_adversary_composition}
	Let $\Pi$ be a coin-flipping protocol, let \Ac be an adaptive adversary for $\Pi$, and let \Bc be an adaptive adversary for $\aprot{\Pi}{\Ac}$. The adversary $\Bc \circ \Ac$ on $\Pi$ is defined as follows:
	\begin{algorithm}[Adversary $\Bc \circ \Ac$ on $\Pi$]~
		\smallskip

		{\bf For} $i \eqdef 1$ {\bf to} $\Length{\Pi}$:

		\quad If $\Cc \in \set{\Ac, \Bc}$ would like to modify the \ith message, corrupt the current party (if not already corrupted), and alter its message according to $\Cc$ (giving priority to $\Bc$ over $\Ac$).
	\end{algorithm}
\end{definition}

It is clear that if \Ac is $k_\Ac$-adaptive and \Bc is $k_\Bc$-adaptive, then $\Bc \circ \Ac$ is $(k_\Ac + k_\Bc)$-adaptive.
\smallskip

\begin{proposition}\label{lem:adaptive_adversary_composition}
	Let $\Pi$, \Ac and \Bc be as in \cref{def:adaptive_adversary_composition}, then $\expecc{\aprot{\Pi}{\Bc \circ \Ac}} = \expecc{\aprot{\brackets{\aprot{\Pi}{\Ac}}}{\Bc}}$. 
\end{proposition}
\begin{proof}
	It is clear that $\aprot{\brackets{\aprot{\Pi}{\Ac}}}{\Bc}$ and $\aprot{\Pi}{\Bc \circ \Ac}$ induce the same distribution on the protocol tree of $\Pi$, and thus induce the same output distribution.
\end{proof}

\smallskip

\smallskip

\paragraph{Composition of strongly adaptive adversaries.}

\begin{definition} [Strongly adaptive adversary composition] \label{def:strongly_adaptive_adversary_composition}
	Let $\Pi$ be a coin-flipping protocol, let \Ac be a strongly adaptive adversary for $\Pi$, and let \Bc be a strongly adaptive adversary for $\aprot{\Pi}{\Ac}$. The adversary $\Bc \circ \Ac$ on $\Pi$ is defined as follows:
	\begin{algorithm}[Adversary $\Bc \circ \Ac$ on $\Pi$]~
		\smallskip

		{\bf For} $i \eqdef 1$ {\bf to} $\Length{\Pi}$:
		\begin{enumerate}
			\item Let $\msg_i$ be the message sent in the \ith round.
			\item Emulate the \ith round of \Ac, with input $\msg_i$. If \Ac wishes to alter $\msg_i$ to $\hatmsg$, corrupt the current party (if not already corrupted), alter the sent message to $\hatmsg$, and set $\msg_i \getsr \hatmsg$.
			\item Emulate the \ith round of \Bc, with input $\msg_i$. If \Bc wishes to alter $\msg_i$ to $\hatmsg$, corrupt the current party (if not already corrupted) and alter the sent message to $\hatmsg$.
		\end{enumerate}
	\end{algorithm}
\end{definition}

It is clear that if \Ac is $k_\Ac$-strongly adaptive and \Bc is $k_\Bc$-strongly adaptive, then $\Bc \circ \Ac$ is $(k_\Ac + k_\Bc)$-strongly adaptive.
\smallskip

\begin{proposition} \label{lem:strongly_adaptive_adversary_composition}
	Let $\Pi$, \Ac and \Bc be as in \cref{def:strongly_adaptive_adversary_composition}, then $\expecc{\aprot{\Pi}{\Bc \circ \Ac}} = \expecc{\aprot{\brackets{\aprot{\Pi}{\Ac}}}{\Bc}}$. 
\end{proposition}
\begin{proof}
	It is clear that $\aprot{\brackets{\aprot{\Pi}{\Ac}}}{\Bc}$ and $\aprot{\Pi}{\Bc \circ \Ac}$ induce the same distribution on the protocol tree of $\Pi$, and thus induce the same output distribution.
\end{proof}

\subsection{Useful Inequalities}
We use the following standard inequalities.
\begin{fact} \label{lem:xlogopx_bound}
	For $-\frac{1}{2} \le x$ it holds that $x\logg{1+x} \le 2x^2$.
\end{fact}

\begin{fact} \label{lem:xlogx_bound2}
	For $0 \le x \le 1$ it holds that $x\log x \ge -1$.
\end{fact}

\newcommand{\epsnt}{\eps_n}
\newcommand{\deltant}{\delta_n}
\newcommand{\lambdant}{\lambda_n}

\def\oneoverepsvalue{\sqrt[50]{\log \log \numparties}}
\def\epsvalue{\nfrac{1}{\oneoverepsvalue}}

\def\amtcorruptions{\bigO{\sqrt{\numparties} \cdot \log \numparties}}

\def\oneoversupportthresholdnt{\lambdant \sqrt{\numparties}}
\def\supportthresholdnt{\nfrac{1}{\oneoversupportthresholdnt}}
\def\variancethresholdnt{\nfrac{1}{\lambdant\numparties}}

\def\oneoversupportthreshold{\lambda \sqrt{\numparties}}
\def\supportthreshold{\nfrac{1}{\oneoversupportthreshold}}
\def\variancethreshold{\nfrac{1}{\lambda\numparties}}

\newcommand{\ycondvar}[1]{\varian{Y_{#1} \cond \CorruptedMsgg{< {#1}}}}

\section{Biasing Robust Coin Flip} \label{sec:AttackingRobustProtocols}
In this section, we present an attack for biasing robust coin-flipping protocols. To simplify notation, we focus on robustness towards $0$; see below. In the following, $n$ typically represents the number of parties of the robust (``non-normal'') protocol.
We make use of the following notation.

\def\epslamdelDef{
	  For $n \in \N$, let $\epsnt \eqdef \epsvalue$, $\lambdant \eqdef \nfrac{100}{\epsnt^5} = 100 \cdot \sqrt[10]{\log \log \numparties}$ and $\deltant \eqdef \nfrac{1}{\log^2 \numparties}$.
}

\begin{notation} \label{notation:epslamdel}
    \epslamdelDef
\end{notation}

The main result of this section is stated below.
\begin{definition}[Robust coin-flipping protocols] \label{def:robustProtocols}
	An $\lenproto$-round, full-information, coin-flipping protocol $\Pi$ is {\sf $(\alpha,\beta)$-robust}, if $\pr{\exists i \in [\lenproto] \colon \minnn{\supp{\protincc{\Pi}{\HonestMsgg{\le i}} \cond \HonestMsgg{< i}}} \le -\alpha} \le \beta$.
\end{definition}

\begin{theorem} [Biasing robust coin-flipping protocols] \label{thm:biasing_zero_robust_protocols}
    Let $\Pi$ be an $\numparties$-party, $(\nfrac{1}{(\lambdant \cdot \sqrt{n})}, \deltant)$-robust, full-information coin-flipping protocol such that $\expecc{\Pi} \ge \epsnt$. Then there exists an $\amtcorruptions$-adaptive adversary \Ac such that $\expecc{\Pi_\Ac} \ge 1 - \epsnt$.
\end{theorem}

We start, \cref{sec:NormalRobustCF}, by proving a variant of \cref{thm:biasing_zero_robust_protocols}
for ``normal'' coin-flipping protocols. Informally, in a \emph{normal} coin-flipping protocol, parties participating in multiple rounds don't have ``too large of an influence'' on the protocol's outcome, though a party sending a single message may greatly influence the outcome. In \cref{sec:NormalRobustCF:4}, we leverage this attack for proving \cref{thm:biasing_zero_robust_protocols} by transforming any given protocol to a normal protocol, and show that the guaranteed attack on the latter protocol yields an attack of essentially the same quality on the original protocol.\footnote{It is worth mentioning that the shift from attacking arbitrary protocols to normal protocols is merely done for notational convince, and nothing exciting is hidden under the hood of the transformation above.}

\subsection{Biasing Normal Robust Coin Flip} \label{sec:NormalRobustCF}

\def\altNP{t}
Normal coin-flipping protocols are coin-flipping protocols of a concrete message-ownership structure. In \cref{sec:NormalRobustCF:4}, we show that an arbitrary coin-flipping protocol can be viewed, with some parameter adaptation, as a normal coin-flipping protocol. 

\begin{notation}[Parties classification] \label{not:normal_protocol}~

	\begin{description}
		\item[Large-jump parties.] A party $\Party$ has {\sf a large jump in $\msg$}, if $\exists i \in \Idx_\Party(\msg)$ s.t. 
		
		$\varian{\protincc{\Pi}{\HonestMsgg{\le i}} \cond \HonestMsgg{< i} = \msgg{< i}} \ge 1/( n \lambdant)$.

		\item[Small-jump parties.] A party $\Party$ has {\sf small jumps in $\msg$} if it participates (sends at least one message) but has no large jumps in $\msg$.

		\item[Unfulfilled parties.] A small-jumps party is {\sf unfulfilled in $\msg$} if 
		
		$\sum_{i \in \Idx_\Party(\msg)}\varian{\protincc{\Pi}{\HonestMsgg{\le i}} \cond \HonestMsgg{< i} = \msgg{< i}} <1/( n \lambdant)$.

	\end{description}
\end{notation}

Similarly, we refer to jumps as small or big jumps, based on their variance, \ie large jumps are jumps of variance $\ge 1/(n\lambdant)$ and small jumps are jumps of variance $< 1/(n\lambdant)$.

\begin{definition} [Normal coin-flipping protocols] \label{def:normal_protocol}
    Let $\Pi$ be a $\altNP$-party, $\lenproto$-round, full-information coin-flipping protocol and let $\numparties \in \N$. We say $\Pi$ is {\sf $\numparties$-normal}, if the following hold: 

    \begin{quote}
    \begin{description}
        \item[Large-jump parties send a single message.] $\size{\Idx_{\Party}(\msg)} = 1$, for every large-jump party \Party in $\msg\in \Supp(\HonestMsg)$.
        
        \item[Small-jump party has bounded overall variance.] \phantom{.}\\
		$\sum_{i \in \Idx_\Party(\msg)} \varian{\protincc{\Pi}{\HonestMsgg{\le i}} \cond \HonestMsgg{< i} = \msgg{< i}} \le 2/( n \lambdant)$, for every small-jumps party \Party in $\msg \in \Supp(\HonestMsg)$.

        \item[Bounded number of unfulfilled parties.] 
        In every transcript $\msg \in \Supp(\HonestMsg)$ there are at most $\numparties$ {\sf unfulfilled} (participating) parties.

        \item[At most one non-robust party.]   
        There exists at most one party $\BadParty \in [\altNP]$ such that     for every    transcript $\msg \in \Supp(\HonestMsg)$ and $i\in [\ell]\setminus \Idx_{\BadParty}(\msg)$:
         \[ 
        \minnn{\supp{\protincc{\Pi}{\HonestMsgg{\le i}} \cond \HonestMsgg{< i} = \msgg{< i}}} > - 1/(\lambdant  \sqrt{n})
         \]
        
        (Namely, only $\BadParty$ can cause large negative jumps.)
    \end{description}
    \end{quote}
\end{definition}
Intuitively, a protocol is $n$-normal if no small-jump parties have too large influence (as a function of $n$).  In \cref{sec:NormalRobustCF:4}, we show how to turn any  $n$-party protocol into a $t$-party, $n$-normal protocol (for some $t > n$).

We present an attack on normal protocols that either corrupts (essentially) \emph{all} messages sent by a party or corrupts \emph{none} of them.
\begin{theorem} [Biasing normal robust coin-flipping protocols] \label{thm:biasing_normal_robust_coinflips}
    For every  $\numparties$-normal, $\lenproto$-round, $(\nfrac{1}{(\lambdant\sqrt{n})}, \deltant)$-robust, full-information coin-flipping protocol $\Pi$ with $\expecc{\Pi} \ge \epsnt$,  there exists  $\amtcorruptions$-adaptive adversary \Ac with $\expecc{\Pi_\Ac} \ge 1 - \epsnt$.
\end{theorem}

That is, if $\Pi$ is $\numparties$-normal and $(\nfrac{1}{(\lambdant\sqrt{n})}, \deltant)$-robust, it can be biased by an $\amtcorruptions$-adaptive adversary.  

\subsubsection{Gently Biasing a Distribution}
We begin by introducing a method of biasing distributions to increase their expectation under some utility function. Our adversary will then use this technique to modify the messages of the corrupted parties, with the utility function being the change induced in the protocol's expected outcome.

\begin{definition} [$\Biased$ distribution] \label{def:properties_biased}
    Let $P$ be a distribution, let $\alpha > 0$ and let $f\colon \supp{P} \mapsto [-\nfrac{1}{\alpha}, \infty)$ be such that $\expecc{\bracketss{f}{P}} = 0$.
    The distribution $\biased{\alpha}{f}{P}$ is defined by $\pr{\biased{\alpha}{f}{P} = e} = \pr{P = e} \cdot \brackets{1 + \alpha \cdot \bracketss{f}{e}}$.
\end{definition}

It is easy to verify this is indeed a distribution.
If $f$ is the identity function, we sometimes omit it from the above notation. We use the following properties of the $\Biased$ distribution (proven in \cref{sec:properties_biased}).

\def\BiasedLemma
{
	For any $P$, $\alpha$ and $f$ as in \cref{def:properties_biased}, it holds that
	\begin{enumerate}
		\item $\expecc{\bracketss{f}{\biased{\alpha}{f}{P}}} = \alpha \cdot \varian{\bracketss{f}{P}}$.
		
		\item $\kld{\biased{\alpha}{f}{P}}{P} \le 2\alpha^2 \cdot \varian{\bracketss{f}{P}}$.
		
		\item $\brackets{p \cdot \biased{\alpha}{f}{P} + \brackets{1 - p} \cdot P} \equiv \biased{p\cdot \alpha}{f}{P}$ for any $p\in [0,1]$.
		
		\item There exist a distribution $(A, B)$ which couples $P$ and $\biased{\alpha}{f}{P}$, \ie $A \equiv P$ and $B \equiv \biased{\alpha}{f}{P}$, such that for any $(a, b) \getsr (A, B)$ it holds that $\bracketss{f}{B} \ge \bracketss{f}{A}$.
	\end{enumerate}
}

\begin{lemma}[Properties of the $\Biased$ distribution]\label{claim:properties_biased}
\BiasedLemma
\end{lemma}
\subsubsection{The Attack}
Using the above tool, we define our attack on normal coin-flipping protocols. 
Fix an $\numparties$-normal, $\lenproto$-round, full-information coin-flipping protocol $\Pi$, such that $\expecc{\Pi} \ge \epsnt$ and $\pr{\Idx_\BadParty(\HonestMsg) \neq \emptyset} \le \deltant$. When clear from the context, we omit the subscript $n$ from the notations $\epsnt, \lambdant, \deltant$. The $\amtcorruptions$-adaptive attacker \Ac on $\Pi$ is defined as follows.

\begin{algorithm}[The attacker \Ac] \label{alg:adversary_robust_normal}
    \item {\bf For} $i \eqdef 1$ {\bf to} $\lenproto$, do the following \emph{before} the \ith message is sent:
    \begin{enumerate}
        \item Let \Party be the party about to send the \ith message. If $\Party = \BadParty$, abort.

        \item Let $\msgg{<i}$ denote the messages sent in the previous rounds. Let $\condMsgg{i}$ be the distribution $\restr{\HonestMsgg{i}}{\HonestMsgg{< i} = \msgg{< i}}$, let $\diff_i(x) \eqdef \protincc{\Pi}{\msgg{< i}, x}$ and let $v_i \eqdef \varian{\bracketss{\diff_i}{\condMsgg{i}}}$.

        \item {\bf If} this is the first message sent by \Party, corrupt \Party according to the following method:
              \begin{enumerate}
                  \item {\bf If} \Party is a large-jump party, \ie $v_i \ge \variancethreshold$, corrupt it with probability $\lambda^2 \cdot \sqrt{v_i}$. \footnote{$\lambda^2 \cdot \sqrt{v_i}$ is indeed in the interval $[0, 1]$: $\Party \neq\BadParty$ implies that $v_i \le \supportthreshold$.} \label{alg:largeparty:corruption}

                  \item {\bf Else} (\Party is a small-jumps party), corrupt \Party with probability $\nfrac{\lambda^2}{\sqrt{\numparties}}$. \label{alg:smallparty:corruption}
              \end{enumerate}

        \item {\bf If} \Party is in the corrupted parties pool: \label{alg:corruptedphase}
              \begin{enumerate}
                  \item {\bf If} \Party is a large-jump party, instruct \Party to broadcast its next message according to $\biased{\nfrac1{\sqrt{v_i}}}{\diff_i}{\condMsgg{i}}$. \label{alg:largeparty:alter}

                  \item {\bf Else}: \label{alg:smallparty:alter}
                        \begin{enumerate}
                            \item {\bf If} $\pr{\Party \text{ is corrupted by } \Ac \cond \Msggg{<i}{\Pi_\Ac} = \msgg{<i}} \le \nfrac{16\lambda^2}{\sqrt{\numparties}}$,\footnote{Recall that $\Pi_\Ac$ is the protocol induced by the attack of \Ac on $\Pi$, and or $\Msggg{< i}{\Pi_\Ac}$ are the first $i-1$ messages in its execution. Hence, we are defining the strategy of \Ac in the \ith round using its strategy in the first $i-1$ rounds, so this self-reference is well defined.}
                            instruct \Party to broadcast its next message according to $\biased{\sqrt{\numparties}}{\diff_i}{\condMsgg{i}}$.

                            \item {\bf Else}, instruct \Party to sample its next message \emph{honestly} (\ie according to $\condMsgg{i}$). \label{alg:adversary:samplehonestly}
                        \end{enumerate}
              \end{enumerate}
    \end{enumerate}
\end{algorithm}

The main difference between the above attacker and its simplified variant presented in \cref{sec:Technique}, is that the above attacker might decide not to modify a message of an already corrupted party (see \stepref{alg:smallparty:alter}). This change enables us to easily bound the KL-divergence between the attacked and all-honest distributions, a bound that plays a critical role in our analysis.\footnote{We are not sure whether this change is mandatory for the attack to go through or merely an artifact of our proof technique that bounds the KL divergence between the attacked and honest execution (see \cref{clm:smallvar}).} In the rest of this section, we analyze the expected outcome of $\aprot{\Pi}{\Ac}$ and the number of parties \Ac corrupts. The proof makes use of the following random variables associated with a random execution of $\Pi_\Ac$.

\begin{notation}[Random variables associated with a random execution of $\Pi_\Ac$]~
	\begin{itemize}
		\item $\CorruptedMsg = (\CorruptedMsgg{1}, \ldots, \CorruptedMsgg{\lenproto})\eqdef \Msggg{}{\Pi_\Ac}$.
        		
		\item $S_k \eqdef \bracketss{\Pi}{\CorruptedMsgg{\le k}}$.

		(Note that $S_0, \ldots, S_\lenproto$ is \emph{sub-martingale} \wrt $\CorruptedMsg$.) 
		
		\item $X_k \eqdef S_k - S_{k-1}$. 
			
		(Note that $X_1, \ldots, X_\lenproto$ are the jumps induced by the attacked execution, \ie $X_i= \protincc{\Pi}{\CorruptedMsgg{\le k}}$. Also note that $S_0 = \expecc{\Pi} \ge \eps$ and $S_\lenproto = S_0 + \sum_{i=1}^\lenproto X_i$.)
		
		\item $\CorruptedParties$: the set of parties corrupted in this execution of \Ac on $\Pi$.
		
		(Note that $\CorruptedParties$ is \emph{not} determined by $\CorruptedMsg$, as there is additional randomness involved.)

        \item $Q_1,\ldots,\condMsgg{\lenproto}$: the value of these variables as computed by \Ac.
	\end{itemize}
\end{notation}

We prove \cref{thm:biasing_normal_robust_coinflips} via three key observations. The first observation, proved in \cref{sec:NormalRobustCF:Coupling}, guarantees a per-round coupling between the change in expected outcome induced by the attack and what would have been the change in an honest execution (conditioned on previous messages).

\def\claimbiasedcoupling
{
	There exists a random variable $Y=(Y_1, \ldots, Y_\lenproto)$ jointly distributed with $\CorruptedMsg$, such that for every $i \in [\lenproto]$: 
	\begin{enumerate}
		\item $X_i \ge Y_i$.
		
		\item Conditioned on $\CorruptedMsgg{< i}$: $Y_i$ is distributed like $\protincc{\Pi}{\restr{\HonestMsgg{\le i}}{\HonestMsgg{< i} = \CorruptedMsgg{< i}}}$, and is independent of $Y_{< i}$ and $\set{\Speaker_i(\CorruptedMsgg{< i}) \in \CorruptedParties}$.
		
	\end{enumerate}
}	
\begin{claim}[Coupling honest and attacked conditional distributions] \label{clm:biased_pairing}
	\claimbiasedcoupling  
\end{claim}

That is, $Y_i$ is distributed like the (conditional) change in expected outcome induced by the \ith step \emph{if it were carried out honestly}, and is never larger than the (conditional) change induced by the \ith step of the attacked execution. 
It is easy to verify that $\expecc{Y_k \cond \CorruptedMsgg{< k}, Y_{< k}} = 0$, \ie $\sum_{i=1}^k Y_i$ is a martingale difference sequence \wrt $(\CorruptedMsgg{k}, Y_{k})$.
For the rest of this section, let $Y$ be the random variable guaranteed by \cref{clm:biased_pairing}.

Next, we consider the set of indices corresponding to robust jumps, defined by
\def\GoodJumpsE{[\lenproto]\setminus \Idx_{\BadParty}(\CorruptedMsg)}
\begin{align}
    \GoodJumps \eqdef \GoodJumpsE
\end{align}

Note that $\GoodJumps$ is a random set, determined by $\CorruptedMsg$.
The following observation (proved in \cref{sec:NormalRobustCF:1}) states that the overall conditional variance of $Y$ contributed by the robust jumps is small, which implies that the variance of $\sum Y_i$ is small. It follows that $\sum Y_i$ 
is typically not ``too small'', and since $X_i \ge Y_i$, that $\sum X_i$ is typically not too small.

\def\claimsmallvar
{
$\expecc{\sum_{i \in \GoodJumps} \ycondvar{i}} < \nfrac{2}{\lambda}$.	
}
\begin{claim}[Bounding $Y$'s conditional variance] \label{clm:smallvar}
	\claimsmallvar   
\end{claim}

Finally, in \cref{sec:NormalRobustCF:2} we prove that attacked execution does not deviate too much, in KL-divergence terms, from the honest execution. This implies that, with overwhelming probability, $\BadParty$ does not participate in the protocol (since it participated in the original protocol with very small probability).

\def\claimsmallkl
{
$\kld{\CorruptedMsg}{\HonestMsg} \le 16^3 \lambda^3$.
}	
\begin{claim}[Bounding KL-Divergence between attacked and honest executions] \label{clm:smallkl}
	\claimsmallkl
\end{claim}

\newcommand{\BadHonest}{\mathrm{Bad}^\Pi}
\newcommand{\BadAttacked}{\mathrm{Bad}^A}
\def\rcorruptions{10 \lambda^4 \sqrt{\numparties}}
Equipped with \cref{clm:biased_pairing,clm:smallvar,clm:smallkl}, we are ready to prove \cref{thm:biasing_normal_robust_coinflips}.
\begin{proof} [Proof of \cref{thm:biasing_normal_robust_coinflips}]~
    \paragraph{Expected outcome.}
    We start by analyzing the expected bias induced by \Ac. Note that

    \begin{align} \label{eq:biasing_robust_normal:varian_expectedoutcome} \textstyle
    	\varian{\sum_{i \in \GoodJumps} Y_i} = \sum_{i=1}^\numparties \varian{Y_i \cdot \ind{i \in \GoodJumps}} = \expecc{\sum_{i \in \GoodJumps} \ycondvar{i}} \le \nfrac{2}{\lambda}
    \end{align}

    The first and second equalities holds by \cref{fact:martingales_pythagoras,fact:total_variance} respectively, since $\expecc{Y_k \cond \CorruptedMsgg{< k}, Y_{< k}} = 0$ and $\ind{i \in \GoodJumps}$ is determined by $\CorruptedMsgg{< k}$. The inequality holds by \cref{clm:smallvar}. 
    
    By \cref{eq:biasing_robust_normal:varian_expectedoutcome} and Chebyshev's inequality, (remember that $\expecc{\sum_{i \in \GoodJumps} Y_i} = 0$)
    \begin{align} \textstyle
        \pr{\sum_{i \in \GoodJumps} Y_i \le \nfrac{-\eps}{2}} \le \pr{|\sum_{i \in \GoodJumps} Y_i| \ge \nfrac{\eps}{2}}
		\le \brackets{\nfrac{\nfrac{\eps}{2}} {\sqrt{2/\lambda}}}^{-2}
        = \nfrac{8}{\lambda \eps^2} = \frac{8}{\eps^2 \cdot \nfrac{100}{\eps^5}} \le \eps / 4 \label{eq:robust_coinflip:ybound}
    \end{align}
    We next show that with overwhelming probability $\GoodJumps = [\lenproto]$, namely $\BadParty$ does not participate in the execution. Let $\BadHonest$ be the event $\set{\Idx_{\BadParty}(\HonestMsg) \neq \emptyset}$, and let $\BadAttacked$ be the event $\set{\Idx_{\BadParty}(\CorruptedMsg) \neq \emptyset}$. By assumption, $\pr{\BadHonest} \le \delta$, and by \cref{clm:smallkl} and data-processing of KL-Divergence,
    \begin{align}
        \kld{\ind{\BadAttacked}}{\ind{\BadHonest}} \le \kld{\CorruptedMsg}{\HonestMsg} \le 16^3 \lambda^3
		\label{ineq:biasing_normal_robust:kld_ub}
    \end{align}
    We also note that,
    \begin{align}
         & \kld{\ind{\BadAttacked}}{\ind{\BadHonest}} = \pr{\BadAttacked} \cdot \logg{\frac{\pr{\BadAttacked}}{\pr{\BadHonest}}} + \brackets{1 - \pr{\BadAttacked}} \cdot \logg{\frac{1 - \pr{\BadAttacked}}{1 - \pr{\BadHonest}}} \label{ineq:biasing_normal_robust:kld_lb} \\
         & = \pr{\BadAttacked} \cdot \logg{\frac{\pr{\BadAttacked}}{\pr{\BadHonest}}} + \brackets{1 - \pr{\BadAttacked}} \cdot \brackets{\logg{1 - \pr{\BadAttacked}} - \logg{1 - \pr{\BadHonest}}} \nonumber \\
         & \ge \pr{\BadAttacked} \cdot \logg{\frac{\pr{\BadAttacked}}{\pr{\BadHonest}}} + (-1 + 0) \ge \pr{\BadAttacked} \cdot \logg{\frac{\pr{\BadAttacked}}{\delta}} - 1. \nonumber
    \end{align}
    where the penultimate inequality follows by \cref{lem:xlogx_bound2}. 
    We {now show} that 
    \begin{align}
		\pr{\BadAttacked} \le \nfrac{\eps}{4} \label{eq:robust_coinflip:notrobustbound}
	\end{align}
   Indeed, assuming \cref{eq:robust_coinflip:notrobustbound} does not hold, then (for sufficiently large $\numparties$), 
    \begin{align}
		\sqrt{\log \log \numparties} \le \frac{\logg{\nfrac{\log^2 \numparties}{4\oneoverepsvalue}}}{4\oneoverepsvalue} - 1 \le \phantom{.} & \kld{\ind{\BadAttacked}}{\ind{\BadHonest}} 
		\label{ineq:biasing_normal_robust:usage_kld_lb} \\
		& \kld{\ind{\BadAttacked}}{\ind{\BadHonest}} \le 16^3 \lambda^3 < \sqrt{\log \log \numparties} \label{ineq:biasing_normal_robust:usage_kld_ub}
    \end{align}
    \Inqref{ineq:biasing_normal_robust:usage_kld_lb} follows by \cref{ineq:biasing_normal_robust:kld_lb}, and \Inqref{ineq:biasing_normal_robust:usage_kld_ub} by \cref{ineq:biasing_normal_robust:kld_ub}. Overall---yielding a contraction. Combining \cref{eq:robust_coinflip:ybound,eq:robust_coinflip:notrobustbound} yields,
    \begin{align*}
         & \expecc{\aprot{\Pi}{\Ac}} = \pr{S_\lenproto = 1} = \pr{S_\lenproto > 0} = \pr{S_0 + \sum_{i=1}^\lenproto X_i > 0} = \pr{\sum_{i=1}^\lenproto X_i > -S_0} \ge \pr{\sum_{i=1}^\lenproto X_i > -\eps} \nonumber \\ 
         & \ge \pr{\textstyle \sum_{i \in \GoodJumps} Y_i > -\eps} - \pr{\Idx_{\BadParty}(\CorruptedMsg) \neq \emptyset} \ge \brackets{1 - \nfrac{\eps}{4}} - \nfrac{\eps}{4} \ge 1 - \nfrac{\eps}{2}. \nonumber
    \end{align*}
    The second inequality follows by \cref{clm:biased_pairing}. The penultimate inequality follows by \cref{eq:robust_coinflip:ybound,eq:robust_coinflip:notrobustbound}.

    \paragraph{Number of corruptions.}
    It is left to argue that \Ac does not perform too many corruptions. We calculate the \emph{expected} number of corruptions, and bound the \emph{actual} number of corruptions using Markov's inequality. We introduce several additional notations. Let $\SmallParties$ and $\LargeParties$ be the (random) sets of small-jumps and large-jump parties (that participate in the execution) \wrt $\CorruptedMsg$, respectively. Let $\SmallJumps \eqdef \set{k \in [\lenproto] \suchthat \Speaker_k(\CorruptedMsg) \in \SmallParties}$ be the set of small jumps, and let $\LargeJumps \eqdef \set{k \in [\lenproto] \suchthat \Speaker_k(\CorruptedMsg) \in \LargeParties}$ be the set of large jumps. Note that all of the above random sets are determined by $\CorruptedMsg$. We first notice that since a small-jumps party is corrupted with probability $\nfrac{\lambda^2}{\sqrt{\numparties}}$, it holds that
    \begin{align}
        \expecc{\size{\SmallParties \cap \CorruptedParties}} = \nfrac{\lambda^2}{\sqrt{\numparties}} \cdot \expecc{\size{\SmallParties}}
    \end{align}

    In addition, the definition of $\numparties$-normal protocols stipulates that for any transcript of $\Pi$, there are at most $\numparties$ unfulfilled parties. Since each fulfilled (not unfulfilled) party contributes at least $\nfrac{1}{\lambda\numparties}$ to the sum of variances, which is small by \cref{clm:smallvar}, we deduce that
    \begin{align}
        \expecc{\size{\SmallParties}} \le 3\numparties
    \end{align}

    Combining the above two observations yields the following bound on the number of corrupted small-jump parties:
    \begin{align*}
        \expecc{\size{\SmallParties \cap \CorruptedParties}} \le \nfrac{\lambda^2}{\sqrt{\numparties}} \cdot 3\numparties = 3\lambda^2\sqrt{\numparties}
    \end{align*}

    As for large-jump parties, for any $k \in [\lenproto]$, partial transcript $t = \msgg{< k}$ and large-jump party \Party sending the \tth{k} message, \Party is corrupted with probability $\lambda^2 \cdot \sqrt{\varian{Y_k \cond \CorruptedMsgg{< k} = \msgg{< k}}} \le \lambda^2 \cdot \sqrt{\lambda n} \cdot \varian{Y_k \cond \CorruptedMsgg{< k} = \msgg{< k}}$. Thus, we have that
    \begin{align}
        \lefteqn{\expecc{\size{\LargeParties \cap \CorruptedParties}}} \\
         & = \expecc{\textstyle \sum_{i \in \LargeJumps} \lambda^2 \cdot\sqrt{\ycondvar{i}}} \le \expecc{\textstyle \sum_{i \in \LargeJumps} \lambda^2 \cdot \sqrt{\lambda\numparties} \cdot \ycondvar{i}}\nonumber \\
         & \le \lambda^3 \cdot \sqrt{\numparties} \cdot \expecc{\textstyle \sum_{i \in \GoodJumps} \ycondvar{i}} \le \lambda^3 \cdot \sqrt{\numparties} \cdot \nfrac{2}{\lambda} = 2\lambda^2 \sqrt{\numparties}.\nonumber
    \end{align}
    The first inequality follows by the definition of a large jump, \ie $\varian{Y_k \cond \CorruptedMsgg{< k} = \msgg{< k}} \ge \nfrac{1}{\lambda n}$, and last inequality by \cref{clm:smallvar}. Therefore, the expected amount of corruptions is at most $5\lambda^2\sqrt{\numparties}$. Hence, by Markov's inequality, with probability at least $1 - \nfrac{\eps}{2}$ the amount of corruptions made by \Ac is at most $\nfrac{10\lambda^3\sqrt{\numparties}}{\eps} < \rcorruptions$. 
    
    \paragraph{Putting it together.}
    Consider the adversary $\Ac'$ that acts just as \Ac, but aborts (letting players continue the execution honestly) once the amount of corruptions surpasses $\rcorruptions = \amtcorruptions$. It holds that
    \begin{align*}
         & \expecc{\aprot{\Pi}{\Ac'}} = \pr{\aprot{\Pi}{\Ac'} = 1} \ge \pr{\aprot{\Pi}{\Ac} = 1 \land \size{\CorruptedParties} < \rcorruptions} \\
         & \ge \pr{\aprot{\Pi}{\Ac} = 1} - \pr{\size{\CorruptedParties} \ge \rcorruptions} \ge 1 - \nfrac{\eps}{2} - \nfrac{\eps}{2} = 1 - \eps,
    \end{align*}
    which concludes the proof of the theorem.
\end{proof}

\def\corruptedevent{C}

\subsubsection[Coupling \texorpdfstring{$X_i$}{Xi} and \texorpdfstring{$\condMsgg{i}$}{\condMsg i}]{Coupling $X_i$ and $\condMsgg{i}$, Proving \cref{clm:biased_pairing}} \label{sec:NormalRobustCF:Coupling}

\begin{claim}[Restatement of \cref{clm:biased_pairing}]
	\claimbiasedcoupling
\end{claim}
\begin{proof}
	Fix $i\in [\lenproto]$ and denote $\Party = \Speaker_i(\CorruptedMsg)$. 
	
	Let $\corruptedevent$ be the event $\braces{\Party \in \CorruptedParties}$.
	Let $\corruptedevent_L$ be the event $\braces{\corruptedevent \land \Party \in \LargeParties}$. 
	Also let $\corruptedevent_S$ be the event ${\ubraces{(\corruptedevent \land \Party \in \SmallParties) \land (\pr{\corruptedevent \cond \CorruptedMsgg{<i} = \msgg{<i}} \le \nfrac{16\lambda^2}{\sqrt{\numparties}})}}$, \ie $\Party$ is small-jump corrupted party and  \Ac instructs  it to alter the current message  (see \stepref{alg:adversary:samplehonestly} of \cref{alg:adversary_robust_normal}). 
	Finally, define the following random variable (determined by $\CorruptedMsgg{< i}$ and $\ind{\corruptedevent}$): 
	\[\alpha = \begin{cases}
		\nfrac{1}{\ycondvar{i}} & \text{if } \ind{\corruptedevent_L} = 1 \\
		\sqrt{n} & \text{if } \ind{\corruptedevent_S} = 1 \\
		0 & \text{otherwise}
	\end{cases}\]

	Now consider the (random) distribution $(A, B)$ guaranteed by \cref{claim:properties_biased}(4) \wrt $P = \condMsgg{i}$, $f = \diff_i$ and $\alpha$.\footnote{($\diff_i(\cdot) = \protincc{\Pi}{\msgg{< i}, \cdot}$)}
	It is easy to verify that by construction (\ie \cref{alg:adversary_robust_normal}): conditioned on $\CorruptedMsgg{< i}$ and $\ind{\corruptedevent}$, $B$ distributes like $\CorruptedMsgg{i}$.

	Now, conditioned on on $\CorruptedMsgg{\le i}$ and $\ind{\corruptedevent}$ we sample $\tMsgg{i} \getsr \restr{A}{B = \CorruptedMsgg{i}}$, independently of $\tMsgg{< i}$.
	
	Finally, we set $Y_i = \diff_i(\tMsgg{i})$, and from the previous observation it immediately follows that $Y_i$ is distributed like $\bracketss{\diff_i}{\restr{\HonestMsgg{\le i}}{\HonestMsgg{< i} = \CorruptedMsgg{< i}}}$.

	It is clear that conditioned on $\CorruptedMsgg{< i}$, $\tMsgg{i}$ is distributed like $Q_i$. In addition, it is independent of $\ind{\corruptedevent}$ because it is distributed the same no matter the value of $\ind{\corruptedevent}$---it is even distributed the same conditioned on $\ind{\corruptedevent}, \ind{\corruptedevent_S}, \ind{\corruptedevent_L}$. And so the same is true for $Y_i$. As for the independence from $Y_{< i}$, it follows immediately by the independence from $\tMsgg{< i}$.
	
	All that is left to show is that $X_i \ge Y_i$.
	\[
		Y_i = \diff_i(\tMsgg{i}) \equiv \diff_i(\restr{A}{B = \CorruptedMsgg{i}}) \le \diff_i(\CorruptedMsgg{i}) = X_i
	\]
	where the inequality follows by the property $f(B) \ge f(A)$ guaranteed by \cref{claim:properties_biased}(4).
\end{proof}

\subsubsection[Bounding \texorpdfstring{$Y$}{Y}'s Conditional Variance]{Bounding $Y$'s Conditional Variance, Proving \cref{clm:smallvar}}\label{sec:NormalRobustCF:1}

\begin{claim}[Restatement of \cref{clm:smallvar}]
	\claimsmallvar
\end{claim}
\begin{proof}
	Immediately follows by \cref{claim:smallvar:largejumps,claim:smallvar:smalljumps}, given below. \cref{claim:smallvar:largejumps} states that $\expecc{\sum_{i} \ycondvar{i}} \le \nfrac{1}{\lambda}$ when $i$ ranges over $\LargeJumps$, and \cref{claim:smallvar:smalljumps} state the same when $i$ ranges over $\SmallJumps$. Hence, \claimsmallvar
\end{proof}

In the following, we use a conditional variant of the biased distribution.
\begin{definition} [Conditional variant of $\Biased$]
	Let $\cY, \cZ$ and $\eta$ be jointly distributed random variables, and let $f\colon \supp{\cY} \mapsto \reals$ be a function, such that (1) $\eta$ is determined by $\cZ$; (2) $\bracketss{f}{\cY} \ge -\nfrac{1}{\eta}$; (3) $\expecc{\bracketss{f}{\cY} \cond \cZ} = 0$. 
	Define the random variable $\cbiased{\eta}{f}{\cY}{\cZ}$, jointly distributed with $\cZ$, by sampling $\cbiased{\eta}{f}{\cY}{\cZ} \getsr \biased{\eta}{f}{\restr{\cY}{\cZ = \cZ}}$.\footnote{
		Here $\restr{\cY}{\cZ = \cZ}$ means we consider the distribution induced by conditioning $\cY$ on the current value of $\cZ$.
	}
\end{definition}

We now move to proving \cref{claim:smallvar:largejumps,claim:smallvar:smalljumps}. 

\paragraph{Large jumps.}
\newcommand{\jumplargeevent}[1]{L_{#1}}
\begin{claim} \label{claim:smallvar:largejumps}
	$\expecc{\sum_{i \in \LargeJumps} \ycondvar{i}} < \nfrac{1}{\lambda}$.
\end{claim}
\begin{proof}
	Let $\jumplargeevent{k}$ be the event $\braces{k \in \LargeJumps}$. Note that $\ind{\jumplargeevent{k}}$ is determined by $\CorruptedMsgg{< k}$. Compute,
	\begin{align}
		& \ind{\jumplargeevent{k}} \cdot \expecc{X_k \cond \CorruptedMsgg{< k}} \label{eq:largejumps} \\
		& = \ind{\jumplargeevent{k}} \cdot \left( \lambda^2 \sqrt{\ycondvar{k}} \cdot \expecc{\cbiased{\nfrac{1}{\sqrt{\ycondvar{k}}}}{}{Y_k}{\CorruptedMsgg{< k}} \cond \CorruptedMsgg{< k}} \label{eq:largejumps:xkcondvar:2} \right. \\
		& \left. \phantom{=====} + \brackets{1 - \lambda^2 \sqrt{\ycondvar{k}}} \cdot \expecc{Y_k \cond \CorruptedMsgg{< k}} \right) \nonumber \\
		& = \ind{\jumplargeevent{k}} \cdot \lambda^2 \sqrt{\ycondvar{k}} \cdot \nfrac{1}{\sqrt{\ycondvar{k}}} \cdot \ycondvar{k} + 0 \label{eq:largejumps:xkcondvar:3} \\
		& = \ind{\jumplargeevent{k}} \cdot \lambda^2 \cdot \ycondvar{k}.\nonumber
	\end{align}
	\Eqref{eq:largejumps:xkcondvar:2} follows by construction (see \stepref{alg:largeparty:corruption} and \stepref{alg:largeparty:alter} of \cref{alg:adversary_robust_normal}) and \cref{clm:biased_pairing}, and \Eqref{eq:largejumps:xkcondvar:3} follows by \cref{claim:properties_biased}(1).
	Hence,
	\begin{align}
		&\textstyle \expecc{\sum_{i=1}^\lenproto X_i} = \sum_{i=1}^\lenproto \expecc{X_i} = \sum_{i=1}^\lenproto \expecc{\expecc{X_i \cond \CorruptedMsgg{< i}}} = \expecc{\sum_{i=1}^\lenproto \expecc{X_i \cond \CorruptedMsgg{< i}}} \label{eq:largejumps:totalexpec:1} \\
		&\ge \expecc{\sum_{i=1}^\lenproto \ind{\jumplargeevent{k}} \cdot \expecc{X_i \cond \CorruptedMsgg{< i}}} \label{eq:largejumps:totalexpec:2} \\
		&= \expecc{\sum_{i=1}^\lenproto \ind{\jumplargeevent{k}} \cdot \lambda^2 \cdot \ycondvar{i}} \label{eq:largejumps:totalexpec:3} \\
		&= \lambda^2 \cdot \expecc{\textstyle \sum_{i \in \LargeJumps} \ycondvar{i}}. \nonumber
	\end{align}
	\Eqref{eq:largejumps:totalexpec:1} follows by the law of total expectation (\cref{fact:total_expectation}), \Inqref{eq:largejumps:totalexpec:2} holds since $\ind{\jumplargeevent{i}}$ is determined by $\CorruptedMsgg{< i}$ and $\expecc{X_i \cond \CorruptedMsgg{< i}} \ge 0$, and \Eqref{eq:largejumps:totalexpec:3} follows by \cref{eq:largejumps}. 

	Thus,
	\[
		\textstyle \lambda^2 \cdot \expecc{\textstyle \sum_{i \in \LargeJumps} \ycondvar{i}} \le \expecc{\sum_{i=1}^\lenproto X_i} \le 1
	\]
	and so $\expecc{\textstyle \sum_{i \in \LargeJumps} \ycondvar{i}} \le \nfrac{1}{\lambda^2} < \nfrac{1}{\lambda}$.
\end{proof}

\def\indiciesparty{\cI}
\paragraph{Small jumps.}
\begin{claim}\label{claim:smallvar:smalljumps}
	$\expecc{\sum_{i \in \SmallJumps} \ycondvar{i}} < \nfrac{1}{\lambda}$.
\end{claim}

For some party \Party denote by $\indiciesparty_\Party = \bracketss{\Idx_\Party}{\CorruptedMsg}$, \ie the indices in which \Party is the speaker. Also, consider the following definition for measuring the contribution of a small-jumps party.

\begin{definition}[Contributional parties]\label{def:Contributional}
	A party \Party is said to be {\sf contributional} if
	\begin{align*} 
		\textstyle \pr{\sum_{i \in \indiciesparty_\Party} \ycondvar{i} > \frac{1}{8\lambda \numparties} \, \mid \Party \in \SmallParties} \ge \nfrac{1}{8}.
	\end{align*}
	Let $\ContribParties$ be the set of contributional parties, and let  $\SmallContribParties \eqdef \SmallParties \cap \ContribParties$.
\end{definition}
Note that being a contributional party is determined by the \emph{protocol}, \ie $\ContribParties$ does not depend on the transcript. In contrast, $\SmallContribParties$, \ie the set of contributional small-jumps parties, does depend on the transcript---because $\SmallParties$ is transcript-dependent. We make use of the following claim.

\begin{claim} \label{claim:contributional_party}
	For any contributional party \Party, it holds that
	\[
		\textstyle \expecc{\sum_{i \in \indiciesparty_\Party} X_i \cond \Party \in \CorruptedParties \cap \SmallContribParties} \ge \nfrac{1}{(256\lambda\cdot \sqrt{\numparties})}.
	\]
\end{claim}

We prove \cref{claim:contributional_party} below, but first use it for proving \cref{claim:smallvar:smalljumps}.

\begin{proof} [Proof of \cref{claim:smallvar:smalljumps}]
We start by using \cref{claim:contributional_party} to lower-bound $\expecc{\sum_{i=1}^\lenproto X_i}$.
\begin{align}
	\lefteqn{\textstyle \expecc{\sum_{i=1}^\lenproto X_i} = \sum_{i=1}^\lenproto \expecc{X_i} = \sum_{i=1}^\lenproto \expecc{\expecc{X_i \cond \CorruptedMsgg{< i}}} = \expecc{\sum_{i=1}^\lenproto \expecc{X_i \cond \CorruptedMsgg{< i}}}} \label{claim:smalljumps:expec_smallcontrib_relation:1} \\
	&\ge \expecc{\textstyle \sum_{\Party \in \SmallContribParties} \sum_{i \in \indiciesparty_\Party} \expecc{X_i \cond \CorruptedMsgg{< i}}} \label{claim:smalljumps:expec_smallcontrib_relation:2} \\
	& = \expecc{\textstyle \sum_{\Party \in \SmallContribParties} \sum_{i \in \indiciesparty_\Party} X_i} \label{claim:smalljumps:expec_smallcontrib_relation:3} \\
	& = \textstyle \mathbb{E} \big[ \sum_{\Party \in \ContribParties} \expecc{\sum_{i \in \indiciesparty_\Party} X_i \cond \Party \in \CorruptedParties \cap \SmallContribParties}  \label{claim:smalljumps:expec_smallcontrib_relation:4} \\
	& \phantom{\textstyle .. \sum_{\Party \in \ContribParties} ==}\cdot \pr{\Party \in \CorruptedParties \cap \SmallContribParties} \big] \nonumber \\
	& \ge  \expecc{\textstyle \sum_{\Party \in \ContribParties} \nfrac{1}{256\lambda\sqrt{\numparties}} \cdot \pr{\Party \in \CorruptedParties \cap \SmallContribParties}} \label{claim:smalljumps:expec_smallcontrib_relation:5} \\
	& = \expecc{\textstyle \sum_{\Party \in \ContribParties} \nfrac{1}{256\lambda\sqrt{\numparties}} \cdot \brackets{\pr{\Party \in \SmallContribParties} \cdot \nfrac{\lambda^2}{\sqrt{\numparties}}}} \label{claim:smalljumps:expec_smallcontrib_relation:6} \\
	&= \nfrac{\lambda}{256\numparties} \cdot \expecc{\size{\SmallContribParties}}. \nonumber
\end{align}
\Inqref{claim:smalljumps:expec_smallcontrib_relation:1} follows by the law of total expectation (\cref{fact:total_expectation}). \Inqref{claim:smalljumps:expec_smallcontrib_relation:2} holds since $\expecc{X_i \cond \CorruptedMsgg{< i}} \ge 0$. \Eqref{claim:smalljumps:expec_smallcontrib_relation:2} follows by the law of total expectation, since $\braces{\Speaker_i(\CorruptedMsg) \in \SmallContribParties}$ is determined by $\CorruptedMsgg{< i}$. \Eqref{claim:smalljumps:expec_smallcontrib_relation:4} also follows by the law of total expectation, combined with the fact that $\expecc{\sum_{i \in \indiciesparty_\Party} X_i \cond \Party \notin \CorruptedParties \land \Party \in \SmallContribParties} = 0$. \Inqref{claim:smalljumps:expec_smallcontrib_relation:5} by \cref{claim:contributional_party}. \Eqref{claim:smalljumps:expec_smallcontrib_relation:6} by construction (see \stepref{alg:smallparty:corruption} of \cref{alg:adversary_robust_normal}).

Moving over from contributional parties, for any \emph{non-contributional} party \Party, it holds that
\[
	\expecc{\textstyle \sum_{i \in \indiciesparty_\Party} \ycondvar{i} \cond \Party \in \SmallParties} < \nfrac{1}{8} \cdot \nfrac{2}{\lambda \numparties} + \nfrac{1}{8\lambda \numparties} < \nfrac{3}{8\lambda \numparties}. \nonumber
\]
The first inequality follows by definition of a small jumps party and that of a non-contributional party, respectively. Finally, by the above inequality,
\begin{align}\label{ineq:SmallNonCon} \textstyle
	\expecc{\sum_{\Party\in \SmallParties \setminus \ContribParties} \sum_{i\in \indiciesparty_\Party} \ycondvar{i}}\le \expecc{\size{\SmallParties}} \cdot \nfrac{3}{8\lambda \numparties}
\end{align}

Denote $\gamma = \expecc{\sum_{i \in \SmallJumps} \ycondvar{i}}$. Each fulfilled (not unfulfilled) party contributes at least $\nfrac{1}{\lambda\numparties}$ to the sum of conditional variances, and so there are at most $\gamma \lambda\numparties$ such parties. Recalling that there are at most $n$ unfulfilled parties, we deduce that 
\begin{align}\label{ineq:SizeSmallParties}
 	\expecc{\size{\SmallParties}} \le \numparties + \numparties \gamma \lambda
\end{align}
Assume towards a contradiction that $\gamma \ge \nfrac{1}{\lambda}$. By \cref{ineq:SizeSmallParties}, $\expecc{\size{\SmallParties}} \le 2\numparties \gamma \lambda$, and thus by \cref{ineq:SmallNonCon}
\begin{align}\label{ineq:NonContribtParties} \textstyle
	\expecc{\sum_{\Party\in \SmallParties \setminus \ContribParties} \sum_{i\in \indiciesparty_\Party} \ycondvar{i}}\le 3\gamma/4
\end{align}
We conclude that $\expecc{\sum_{\Party\in \SmallContribParties} \sum_{i\in \indiciesparty_\Party} \ycondvar{i}}\ge \gamma/4$, and since by definition, every small jumps party contributes at most $\nfrac{2}{\lambda n}$ to the sum of conditional variances, we deduce that 
\begin{align}\label{ineq:SizeSmallContribParties}
\expecc{\size{\SmallContribParties}}\ge \gamma \lambda n/8
\end{align}
Combining \cref{claim:smalljumps:expec_smallcontrib_relation:1,ineq:SizeSmallContribParties}, yields that
\begin{align} 
	   \expecc{\sum_{i=1}^\lenproto X_i} \ge \nfrac{\gamma\lambda^2}{2048}
\end{align}
It follows that $\expecc{S_\lenproto} = \expecc{S_0 + \sum_{i=1}^\lenproto X_i} \ge \eps + \nfrac{\gamma\lambda^2}{2048}$, which is larger than $1$ for sufficiently large $\numparties$, yielding a contradiction.  Hence, $\expecc{\sum_{i \in \SmallJumps} \ycondvar{i}} = \gamma < \nfrac{1}{\lambda}$, concluding the proof. 
\end{proof}

\paragraph{Proving \cref{claim:contributional_party}.}
\begin{proof} [Proof of \cref{claim:contributional_party}]
	Fix a contributional party \Party, and consider the following events (jointly distributed with $\CorruptedMsg$):
	\begin{itemize}
		\item $C = \braces{\Party \in \CorruptedParties}$.
		
		\item $S = \braces{\Party \in \SmallParties}$. 
		
		\item $L = \braces{\sum_{i \in \indiciesparty_\Party} \ycondvar{i} > \nicefrac{1}{8\lambda\numparties}}$, \ie, \Party has large conditional variance.
		
		\item $H = \braces{\forall k \in \indiciesparty_\Party \suchthat \pr{\Party \in \CorruptedParties \cond \CorruptedMsgg{< k}} < 16 \cdot \nfrac{\lambda^2}{\sqrt{\numparties}}}$, \ie \stepref{alg:adversary:samplehonestly} of \cref{alg:adversary_robust_normal} never happens for \Party.
	\end{itemize}
	We start by proving that $\pr{H \land L \cond S}$ is large, and then use a KL-divergence argument to deduce that $\pr{H \land L \cond S \land C}$ is large, \ie \Party encounters large conditional variance even when it is a corrupted small-jumps party. This will imply that the change in expectation \Party induces (when a corrupted small-jumps party) is large.
	
	To prove that $\pr{H \land L \cond S}$ is large, we move to the conditional probability space where $S$ occurs (\ie \Party participates in the protocol as a small-jumps party). Consider the martingale $C_0, \ldots, C_\lenproto$ defined by $C_k \eqdef \expecc{\ind{C} \cond \CorruptedMsgg{\le k}}$ (that is, $C_k$ is the projection of the event $C$ on the information held by $\CorruptedMsgg{\le k}$). Under the conditioning \Party is a small-jumps party, therefore, the adversary corrupts \Party with probability $\nfrac{\lambda^2}{\sqrt{\numparties}}$, \ie $\expecc{\ind{C}} = \pr{C} = \nfrac{\lambda^2}{\sqrt{\numparties}}$. Thus by Doob's maximal inequality (see \cref{fact:doob_maximal_inequality}), it holds that
	\begin{equation}
		\pr{\neg H} = \pr{\sup \braces{C_0, \ldots, C_\ell} \ge 16 \cdot \nfrac{\lambda^2}{\sqrt{\numparties}}} \le \nfrac{1}{16}
	\end{equation}
	
	Back to the non-conditional probability space, we deduce that 
	\begin{align} \label{eq:claim:contrib:1}
		&\pr{L \land H \cond S} \ge \pr{L \cond S} - \pr{\neg H \cond S} \ge \nfrac{1}{8} - \nfrac{1}{16} = \nfrac{1}{16}
	\end{align}
	where $\pr{L \cond S} \ge \nfrac{1}{8}$ holds since \Party is contributional. We next bound $\kld{\restr{\CorruptedMsg}{S \land C}}{\restr{\CorruptedMsg}{S}}$. Letting $\tMsg$ denote the distribution $\restr{\CorruptedMsg}{S \land C}$ compute,  
	\begin{align}\label{eq:claim:contrib:2}
		\lefteqn{ \kld{\restr{\CorruptedMsg}{S \land C}}{\restr{\CorruptedMsg}{S}} = \kld{\tMsg}{\restr{\CorruptedMsg}{S}}} \\
		& = \sum_{i=1}^\lenproto \texpec{\msg \getsr \tMsg}{\kld{\restr{\CorruptedMsgg{i}}{\CorruptedMsgg{< i} = \msgg{< i} \land S \land C}}{\restr{\CorruptedMsgg{i}}{\CorruptedMsgg{< i} = \msgg{< i} \land S}}} \label{eq:claim:contrib:2:1} \\
		& = \texpec{\msg \getsr \tMsg}{\textstyle \sum_{i \in \Idx_\Party(\msg)} \kld{\restr{\CorruptedMsgg{i}}{\CorruptedMsgg{< i} = \msgg{< i} \land C}}{\restr{\CorruptedMsgg{i}}{\CorruptedMsgg{< i} = \msgg{< i}}}} \label{eq:smalljumps:kl2} \\
		& \le \texpec{\msg \getsr \tMsg}{\textstyle \sum_{i \in \Idx_\Party(\msg)} \kld{\restr{\CorruptedMsgg{i}}{\CorruptedMsgg{< i} = \msgg{< i} \land C}}{\restr{\CorruptedMsgg{i}}{\CorruptedMsgg{< i} = \msgg{< i} \land \neg{C}}}} \label{eq:smalljumps:kl3} \\
		& \le \texpec{\msg \getsr \tMsg}{\textstyle \sum_{i \in \Idx_\Party(\msg)} \kld{\biased{\sqrt{\numparties}}{\protincc{\Pi}{\msgg{< i}, \; \cdot}}{\restr{\HonestMsgg{i}}{\HonestMsgg{< i} = \msgg{< i}}}}{\restr{\HonestMsgg{i}}{\HonestMsgg{< i} = \msgg{< i}}}} \label{eq:smalljumps:kl4} \\
		& \le \texpec{\msg \getsr \tMsg}{\textstyle \sum_{i \in \Idx_\Party(\msg)} 2\numparties \cdot \varian{\protincc{\Pi}{\HonestMsgg{\le i}} \cond \HonestMsgg{< i} = \msgg{< i}}} \label{eq:smalljumps:kl5} \\
		& = 2\numparties \cdot \texpec{\msg \getsr \tMsg}{\textstyle \sum_{i \in \Idx_\Party(\msg)} \varian{\protincc{\Pi}{\HonestMsgg{\le i}} \cond \HonestMsgg{< i} = \msgg{< i}}} \\
		& \le 2\numparties \cdot \expecc{\nfrac{2}{\lambda\numparties}} \le \nfrac{4}{\lambda}. \nonumber
	\end{align}

	\Eqref{eq:claim:contrib:2:1} follows by chain-rule of KL Divergence (see \cref{fact:kl_chain_rule}). \Eqref{eq:smalljumps:kl2} follows by the fact that the conditional distribution of messages not sent by \Party is not affected by conditioning on $C$, and we can drop the conditioning on $S$ since it is determined by $\msgg{< i}$. \Inqref{eq:smalljumps:kl3} follows by convexity of KL-Divergence (\cref{fact:kl_convex}) ($\CorruptedMsgg{i}$ is a convex combination of $\restr{\CorruptedMsgg{i}}{C}$ and $\restr{\CorruptedMsgg{i}}{\neg C}$). \Inqref{eq:smalljumps:kl4} follows by construction (see \stepref{alg:smallparty:alter} of \cref{alg:adversary_robust_normal}) and the convexity of KL-Divergence---the altered messages are a convex combination of honest messages and biased messages (caused by \stepref{alg:adversary:samplehonestly} of \cref{alg:adversary_robust_normal}). \Inqref{eq:smalljumps:kl5} follows from \cref{claim:properties_biased}(2), and the penultimate inequality holds since, by assumption, the protocol $\Pi$ is $\numparties$-normal.
	
	By \cref{eq:claim:contrib:2} and the Pinsker bound (see \cref{fact:pinsker_bound}), it holds that $\sdist{\restr{\CorruptedMsg}{S \land C}}{\restr{\CorruptedMsg}{S}} \le \nfrac{2}{\sqrt{\lambda}}$. Consequently, by \cref{eq:claim:contrib:1} and the data-processing inequality of statistical distance (\cref{fact:dataprocess}), it holds that (for sufficiently large $\numparties$)
	\begin{align}
		\pr{H \land L \cond S \land C} \ge \pr{H \land L \cond S} - \nfrac{2}{\sqrt\lambda} = \nfrac{1}{16} - \nfrac{2}{\sqrt\lambda} > \nfrac{1}{32} \label{eq:smalljumps:good_event_highprob}
	\end{align}
	
	In other words, even when \Party is a corrupted small-jumps party, it still encounters large conditional variance and biases all jumps it encounters. Therefore, all that is left to do is analyze the expectancy of \Party's increments under this conditioning. For $\msg \in \supp{\CorruptedMsg}$ let $\ind{H}(\msg)$ be \emph{the value} of $\ind{H}$ as determined by $\CorruptedMsg = \msg$. Compute, 
	\begin{align}
		{\textstyle \expecc{\sum_{i \in \indiciesparty_\Party} X_i \cond S \land C}} 
		& = \texpec{\msg \getsr \tMsg}{\sum_{i \in \Idx_\Party(\msg)} \expecc{X_i \cond \CorruptedMsgg{< i} = \msgg{< i} \land C}} \nonumber \\
		& \ge \texpec{\msg \getsr \tMsg}{\ind{H}(\msg) \cdot \sum_{i \in \Idx_\Party(\msg)} \expecc{\biased{\sqrt{\numparties}}{}{\restr{Y_i}{\CorruptedMsgg{< i} = \msgg{< i}}}}} \label{eq:smalljumps:2ex1} \\
		& = \expecc{\ind{H} \cdot \textstyle \sum_{i \in \indiciesparty_\Party} \sqrt{\numparties} \cdot \ycondvar{i} \cond S \land C}\label{eq:smalljumps:2ex1:5}\\
		& \ge \expecc{\ind{H} \cdot \sqrt{\numparties} \cdot \ind{L} \cdot \nfrac{1}{8\lambda\numparties} \cond S \land C} \label{eq:smalljumps:2ex2} \\
		& = \nfrac{1}{8\lambda\sqrt{\numparties}} \cdot \pr{H \land L \cond S \land C} \ge \nfrac{1}{8\lambda\sqrt{\numparties}} \cdot \nfrac{1}{32} = \nfrac{1}{256 \lambda \sqrt{\numparties}}. \nonumber
	\end{align}
	\Inqref{eq:smalljumps:2ex1} follows by the definition of $\Ac$ (see \stepref{alg:smallparty:alter} of \cref{alg:adversary_robust_normal}) and \cref{clm:biased_pairing} ($Y_i$ is independent of $C$ conditioned on $\CorruptedMsgg{< i}$). \Eqref{eq:smalljumps:2ex1:5} follows from \cref{claim:properties_biased}(1). \Inqref{eq:smalljumps:2ex2} follows by a point-wise inequality, and the last inequality follows by \cref{eq:smalljumps:good_event_highprob}.
\end{proof}
\subsubsection[Bounding KL-Divergence between Attacked and Honest Executions ]{Bounding KL-Divergence between Attacked and Honest Executions, Proving \cref{clm:smallkl}}\label{sec:NormalRobustCF:2}

\begin{claim} [Restatement of \cref{clm:smallkl}]
	 \claimsmallkl
\end{claim}

The core of the proof relies on \cref{claim:properties_biased}(3) that states that corrupting some party with probability $p$ and then biasing its message according to $\Biased_{\alpha}^f$, is equivalent to biasing this message according to $\Biased_{p\alpha}^f$. This fact yields the following observation:

\begin{claim} \label{claim:smallkl:single}
	For any $i \in [\lenproto]$ and $\msgg{< i} \in \Supp(\CorruptedMsgg{< i})$, it holds that 
	\begin{align*}
		\kld{\restr{\CorruptedMsgg{i}}{\CorruptedMsgg{< i} = \msgg{< i}}}{\restr{\HonestMsgg{i}}{\HonestMsgg{< i} = \msgg{< i}}} \le 16^3 \lambda^4 \cdot \varian{\protincc{\Pi}{\HonestMsgg{\le i}} \cond \HonestMsgg{< i} = \msgg{< i}}.
	\end{align*}
\end{claim}

\begin{proof} 
	Let $\Party \eqdef \party(\msgg{< i})$ (\ie be the party sending the \ith message).
	If \Party is $\BadParty$, we are done since its messages are unchanged (never corrupted). Otherwise, we separately deal with the case that \Party is a small-jumps party and a large-jump party (as determined by the partial transcript $\msgg{< i}$).
	
	\begin{description}
		\item[\Party is a small-jumps party.]
		Conditioned on ${\CorruptedMsgg{< i} = \msgg{< i}}$, the \ith message is altered from its honest (conditional) distribution according to $\Pi$, with probability $p \le \nfrac{16\lambda^2}{\sqrt{\numparties}}$ (see \stepref{alg:smallparty:alter} of \cref{alg:adversary_robust_normal}).
		If the \ith message is altered, it is sampled according to $\biased{\sqrt{\numparties}}{\protincc{\Pi}{\msgg{< i}, \; \cdot}}{\restr{\HonestMsgg{i}}{\HonestMsgg{< i} = \msgg{< i}}}$. By \cref{claim:properties_biased}(3), $\restr{\CorruptedMsgg{i}}{\CorruptedMsgg{< i} = \msgg{< i}}$ is distributed like $\biased{p \sqrt{\numparties}}{\protincc{\Pi}{\msgg{< i}, \; \cdot}}{\restr{\HonestMsgg{i}}{\HonestMsgg{< i} = \msgg{< i}}}$. Hence, by \cref{claim:properties_biased}(2) 
		\begin{align*}
			&\kld{\restr{\CorruptedMsgg{i}}{\CorruptedMsgg{< i} = \msgg{< i}}}{\restr{\HonestMsgg{i}}{\HonestMsgg{< i} = \msgg{< i}}} \\
			&\le 2 \cdot \brackets{p \sqrt{\numparties}}^2 \cdot \varian{\protincc{\Pi}{\HonestMsgg{\le i}} \cond \HonestMsgg{< i} = \msgg{< i}} \\
			&\le 2 \cdot 16^2\lambda^4 \cdot \varian{\protincc{\Pi}{\HonestMsgg{\le i}} \cond \HonestMsgg{< i} = \msgg{< i}}.
		\end{align*}
		
		\item[\Party is a large-jump party.]
		Conditioned on ${\CorruptedMsgg{< i} =\msgg{< i}}$, the \ith message is altered from its honest (conditional) distribution according to $\Pi$, with probability $\lambda^2 \sqrt{v}$ where $v \eqdef \varian{\protincc{\Pi}{\HonestMsgg{\le i}} \cond \HonestMsgg{< i} = \msgg{< i}}$. If the \ith message is altered, it is sampled according to $\biased{\nfrac{1}{\sqrt{v}}}{\protincc{\Pi}{\msgg{< i}, \; \cdot}}{\restr{\HonestMsgg{i}}{\HonestMsgg{< i} = \msgg{< i}}}$. Hence, by \cref{claim:properties_biased}(3), $\restr{\CorruptedMsgg{i}}{\CorruptedMsgg{< i} = \msgg{< i}}$ is distributed like $\biased{\lambda^2}{\protincc{\Pi}{\msgg{< i}, \; \cdot}}{\restr{\HonestMsgg{i}}{\HonestMsgg{< i} = \msgg{< i}}}$. By \cref{claim:properties_biased}(2), we conclude that
		\begin{equation*}
			\kld{\restr{\CorruptedMsgg{i}}{\CorruptedMsgg{< i} = \msgg{< i}}}{\restr{\HonestMsgg{i}}{\HonestMsgg{< i} = \msgg{< i}}} 
			\le 2 \cdot \lambda^4 \cdot \varian{\protincc{\Pi}{\HonestMsgg{\le i}} \cond \HonestMsgg{< i} = \msgg{< i}}.
		\end{equation*}
	\end{description}
\end{proof}

\begin{proof} [Proof of \cref{clm:smallkl}]
	Let the set $\GoodJumps(\msg)$ denote the {value} of $\GoodJumps$ determined by $\CorruptedMsg = \msg$. Compute,
	\begin{align}
		&\kld{\CorruptedMsgg{\le \lenproto}}{\HonestMsgg{\le \lenproto}} \nonumber \\
		&= \sum_{i=1}^\lenproto \texpec{\msg \getsr \CorruptedMsg}{\kld{\restr{\CorruptedMsgg{i}}{\CorruptedMsgg{< i} = \msgg{< i}}}{\restr{\HonestMsgg{i}}{\HonestMsgg{< i} = \msgg{< i}}}} \label{eq:smallkl1} \\
		&= \texpec{\msg \getsr \CorruptedMsg}{\sum_{i \in \GoodJumps(\msg)} \kld{\restr{\CorruptedMsgg{i}}{\CorruptedMsgg{< i} = \msgg{< i}}}{\restr{\HonestMsgg{i}}{\HonestMsgg{< i} = \msgg{< i}}}} \label{eq:smallkl2} \\
		&\le \texpec{\msg \getsr \CorruptedMsg}{\sum_{i \in \GoodJumps(\msg)} 16^3\lambda^4 \cdot \varian{\protincc{\Pi}{\HonestMsgg{\le i}} \cond \HonestMsgg{< i} = \msgg{< i}}} \label{eq:smallkl3} \\
		&= 16^3\lambda^4 \cdot \expecc{\textstyle \sum_{i \in \GoodJumps} \varian{\protincc{\Pi}{\HonestMsgg{\le i}} \cond \HonestMsgg{< i} = \CorruptedMsgg{< i}}} \nonumber \\
		&= 16^3\lambda^4 \cdot \expecc{\textstyle \sum_{i \in \GoodJumps} \ycondvar{i}} \label{eq:smallkl4} \\
		&\le 16^3 \lambda^3. \label{eq:smallkl5}
	\end{align}
	\Eqref{eq:smallkl1} follows by chain rule of KL Divergence (see \cref{fact:kl_chain_rule}). \Eqref{eq:smallkl2} follows since non-$\GoodJumps$ are not corrupted. \Inqref{eq:smallkl3} follows by \cref{claim:smallkl:single}. \Inqref{eq:smallkl4} follows by definition of $Y_i$ (\ie its conditional distribution is $\protincc{\Pi}{\HonestMsgg{\le i}}$). And finally, \Inqref{eq:smallkl5} follows by \cref{clm:smallvar}.
\end{proof}
\newcommand{\smallP}{{\mathsf{small}}}
\newcommand{\largeP}{{\mathsf{large}}}
\subsection{Biasing Robust Coin Flip}\label{sec:NormalRobustCF:4}
In this section, we use the attack on normal robust protocols proved to exist in \cref{sec:NormalRobustCF}, for attacking \emph{arbitrary} robust protocols. We do that by transforming an arbitrary robust protocol into a related normal coin-flipping protocol and proving that the attack on the latter normal protocol stated in \cref{thm:biasing_normal_robust_coinflips}, yields an attack of essentially the same quality on the original (non-normal) protocol, thus proving \cref{thm:biasing_zero_robust_protocols}.

We start by defining the normal form variant of a coin-flipping protocol. Let $\Pi$ be an $\numparties$-party, $\lenproto$-round, full-information coin-flipping protocol. Letting $\altNP = 2 \lenproto \numparties + 1$, the $\altNP$-party, $\lenproto$-round, $\numparties$-normal variant of $\Pi$, is defined as follows:

\begin{protocol}[$\numparties$-normal protocol $\tPi$]\label{proto:NormalTrans}~
	\begin{enumerate}
		\item For each party $\Pc$ of the protocol $\Pi$, the protocol $\tPi$ has $2\lenproto$ parties $\Pc^\smallP_1, \ldots, \Pc^\smallP_\lenproto$ and $\Pc^\largeP_1, \ldots, \Pc^\largeP_\lenproto$. In addition, $\tPi$ has a special party named $\BadParty$.
		
		\item For each party $\Pc$ of $\Pi$, start three counters $L_\Pc = S_\Pc =1$, and $A_\Pc=0$.
		
		\item In rounds $i=1$ to $\lenproto$, the protocol is defined as follows.
		\begin{enumerate}
			\item Let $\msgg{<i}$ denote the messages sent in the previous rounds, and let $\Pc$ be the party that would have sent the \ith message in $\Pi$ given this transcript.
			
			\item Let $\condMsgg{i}$ be the distribution $\restr{\HonestMsgg{i}}{\HonestMsgg{< i} = \msgg{< i}}$ and let $v_i \eqdef \varian{\protincc{\Pi}{\msgg{< i}, \condMsgg{i}}}$.
			
			\item Set $\Party'$ (the ``active'' party) as follows:
			
			\begin{enumerate}
				\item {\bf If} $\minnn{\supp{\protincc{\Pi}{\msgg{< i,} \condMsgg{i}}}} \le -\supportthresholdnt$, set $\Party'$ to $\BadParty$. \label{alg:badparty}
				
				\item {\bf Else, If} $v_i \ge \variancethresholdnt$, set $\Party'$ to $\Pc^\largeP_{L_\Pc}$, and update $L_\Pc = L_\Pc+1$. \label{alg:bigjumps}
				
				\item {\bf Else, If} $v_i < \variancethresholdnt$: \label{alg:smalljumps}
				
				\quad Set $\Party'$ to $\Pc^\smallP_{S_\Pc}$ and update $A_\Pc = A_\Pc + v_i$
				
				\quad {\bf If} $A_\Pc > \variancethresholdnt$:
				
				\begin{itemize}[leftmargin=1.3cm]
					\item Set $S_\Pc = S_\Pc+1$.
					\item Set $A_\Pc = 0$.
				\end{itemize}
			\end{enumerate}
			
			\item $\Party'$ sends the \ith message, as $\Pc$ would in $\Pi$ given the partial transcript $\msgg{<i}$.
		\end{enumerate}
	\end{enumerate}
	
\end{protocol}

\begin{claim}\label{claim:NormalTrans}
	Assume $\Pi$ is a $\numparties$-party full-information coin-flipping protocol, then $\tPi$ is a $\numparties$-normal full-information coin-flipping protocol.
\end{claim}
\begin{proof}
	We handle each of the conditions independently,
	\begin{description}
		\item[Single non-robust party:] \stepref{alg:badparty} properly handles jumps that should belong to $\BadParty$.
		
		\item[Large-jump party sends a single message:] Clearly  \stepref{alg:bigjumps} associates parties of the form $\Party_k^\largeP$ with at most one jump, and it is clear that only parties of this form might have large jumps.
		
		\item[Small-jumps party has bounded overall variance:] \stepref{alg:smalljumps} assures that once $A_\Pc > \variancethresholdnt$, namely the active party has a sum of conditional variances which is larger than $\variancethresholdnt$, it is never associated with another jump further along the execution. Thus, since $A_\Pc$ increases by at most $\variancethresholdnt$ at a time, it never surpasses $2 \cdot \variancethresholdnt$.
		
		\item[At most $\numparties$ unfulfilled parties:] Note that parties which have a sum of conditional variances which is at most $\variancethresholdnt$ must be parties of the form $\Party_k^\smallP$, and the only parties of this form that participate in the protocol (namely, unfulfilled parties) are $\Party_{S^f_\Party}^\smallP$ where \Party is some party (at most $\numparties$) and $S^f_\Party$ is the final value of $S_\Party$. Therefore, at most $\numparties$ unfulfilled parties exist for any transcript.
	\end{description}
\end{proof}

\newcommand{\tA}{\MathAlgX{\widetilde{\Ac}}}

\paragraph{Proving \cref{thm:biasing_zero_robust_protocols}.} 
Given the above tool and \cref{thm:biasing_normal_robust_coinflips}, the proof of \cref{thm:biasing_zero_robust_protocols} is immediate.
\begin{proof}[Proof of \cref{thm:biasing_zero_robust_protocols}]
	Let $\tPi$ be the $\numparties$-normal variant of $\Pi$ defined by \cref{proto:NormalTrans}. By \cref{thm:biasing_normal_robust_coinflips}, there exists a $\amtcorruptions$-adaptive adversary \tA for $\tPi$ such that $\expecc{\aprot{\tPi}{\tA}} \ge 1 - \epsnt$. Consider the adversary \Ac on $\Pi$ that emulates \tA while transforming corruptions of the parties of $\tPi$ to parties of $\Pi$ according to the mapping implicitly defined in \cref{proto:NormalTrans}.
	It is clear that $\expecc{\aprot{\Pi}{\Ac}} = \expecc{\aprot{\tPi}{\tA}} \ge 1 - \epsnt$. In addition, corrupting $k$ parties in $\tPi$ is translated to corrupting at most $k$ parties of $\Pi$, since by construction the parties in $\tPi$ are refinements of the parties in $\Pi$. We conclude that \Ac is the desired $\amtcorruptions$-adaptive 
\end{proof}

\subsection{Proving \cref{claim:properties_biased}}\label{sec:properties_biased}
In this section, we prove \cref{claim:properties_biased}.

\begin{lemma}[Restatement of \cref{claim:properties_biased}]
	\BiasedLemma
\end{lemma}

\begin{proof}[Proof of \cref{claim:properties_biased}]~
	\begin{description}
		\item[Item 1:]
		\begin{align*}
			& \expecc{\bracketss{f}{\biased{\alpha}{f}{P}}} = \sum_{e \in \supp{P}} f(e) \cdot \pr{\biased{\alpha}{f}{P} = e} \\
			& = \sum_{e \in \supp{P}} f(e) \cdot \pr{P = e} \cdot (1 + \alpha f(e)) \\
			& = \expecc{f(P) \cdot (1 + \alpha f(P))} = \expecc{f(P)} + \alpha \cdot \expecc{f^2(P)} = \alpha \cdot \varian{f(P)}.
		\end{align*}
		
		\item[Item 2:]
		\begin{align*}
			& \kld{\biased{\alpha}{f}{P}}{P} = \sum_{e \in \supp{P}} \pr{\biased{\alpha}{f}{P} = e} \cdot \logg{\frac{\pr{\biased{\alpha}{f}{P} = e}}{\pr{P = e}}} \\
			& = \sum_{e \in \supp{P}} \pr{P = e} \cdot (1 + \alpha f(e)) \cdot \logg{1 + \alpha f(e)} \\
			& = \expecc{(1 + \alpha f(P)) \cdot \logg{1 + \alpha f(P)}} = \expecc{\logg{1 + \alpha f(P)}} + \expecc{\alpha f(P) \cdot \logg{1 + \alpha f(P)}} \nonumber \\
			& \le \logg{1 + \expecc{\alpha f(P)}} + \expecc{2\alpha^2 f^2(P)} = 2\alpha^2 \cdot \varian{f(P)}.
		\end{align*}
		The last inequality follows by Jensen's inequality and \cref{lem:xlogopx_bound}.
		
		\item[Item 3:]
		\begin{align*}
			& \pr{\brackets{p \cdot \biased{\alpha}{f}{P} + \brackets{1 - p} \cdot P} = e} \\
			& = p \cdot \pr{\biased{\alpha}{f}{P} = e} + \brackets{1 - p} \cdot \pr{P = e} \\
			& = p \cdot \pr{P = e} \cdot (1 + \alpha f(P)) + \brackets{1-p} \cdot \pr{P = e} \\
			& = \pr{P = e} \cdot \brackets{1 + p\alpha f(P)}.
		\end{align*}
		
		\item[Item 4:] 
		Consider the following random process: Sample $a \gets P$. If $f(a) \ge 0$, set $b = a$. If $f(a) < 0$ with probability $1 + \alpha f(a)$ set $b = a$, otherwise sample $b \gets P^+_f$ for
		\begin{align*}
			P^+_f \equiv \begin{cases}
				e \text{ with probability } \frac{\pr{P=e} \cdot f(e)}{\expecc{\abs{f(P)}}} & \quad \text{ for } e \in \supp{P} $ with $ f(e) > 0
			\end{cases}
		\end{align*}
		By construction $f(b) \ge f(a)$. In addition, it is not hard to verify that the marginal distributions of $a$ and $b$ are that of $P$ and $\biased{\alpha}{f}{P}$, respectively.
	\end{description}
\end{proof}

\newcommand{\advrep}{\Ac^\mathrm{R}_0}
\def\corruptionsrequired{\nfrac{\sqrt{\numparties}}{\lambda\delta}}

\def\restatementepslamdel{
	\begin{notation} [Restatement of \cref{notation:epslamdel}]
		\epslamdelDef
	\end{notation}
}

\section{Biasing Arbitrary Coin Flip}\label{sec:arbitrary_coinflip}
In this section, we use the attack on robust protocols, described in \cref{sec:AttackingRobustProtocols}, to prove our main result: an adaptive attack on any full-information coin-flipping protocols. The main result of our paper is given below. Recalling our notations,

\restatementepslamdel

\begin{theorem} [Biasing full-information coin flips] \label{thm:main_result}
    For any $\numparties$-party, full-information coin-flipping protocol $\Pi$, there exists a $\bigO{\sqrt{\numparties} \cdot \log^3 \numparties}$-adaptive adversary \Ac, such that $\expecc{\aprot{\Pi}{\Ac}} \le \epsnt$ or $\expecc{\aprot{\Pi}{\Ac}} \ge 1 - \epsnt$.
\end{theorem}
Our proof makes use of the following deterministic one-shot (modifies at most a single message) adversary attacking an $\numparties$-party full-information coin-flipping protocol $\Gamma$. The adversary takes advantage of large negative jumps in order to bias the $\Gamma$'s output towards $0$, 

\begin{algorithm}[One-shot adaptive adversary $\Bc$ on $\Gamma$] \label{alg:adversary_badjumps}	
    \item {\bf For} $i =1$ {\bf to} $\Length{\Gamma}$:
    \begin{enumerate}
        \item Let $\msgg{<i}$ be the messages sent in the previous rounds, and let $\partition$ be the party about to send the \ith message.

        \item Denote $\cM_i \eqdef \supp{\Msggg{\le i}{\Gamma} \cond \Msggg{< i}{\Gamma} = \msgg{< i}}$,
        
        {\bf If} no message was corrupted before and $\exists m \in \cM_i \suchthat \protincc{\Gamma}{\msgg{< i}, m} \le -\supportthresholdnt$,
        corrupt and instruct $\partition$ to broadcast such a message $m$ as it next message.
    \end{enumerate}
\end{algorithm}

\smallskip
\smallskip
The proof of the following fact is immediate.
\begin{claim}\label{claim:arbitrarycoinflip:expec_oneshot}
    Let $\Gamma$ be an $\numparties$-party full-information coin-flipping protocol. Then:
    \begin{align*}
        \expecc{\aprot{\Gamma}{\Bc}} \ge \expecc{\Gamma} + \supportthresholdnt \cdot \pr{\exists i \suchthat \minnn{{\supp{\protincc{\Gamma}{\Msggg{\le i}{\Gamma}} \cond \Msggg{< i}{\Gamma}}}} \le -\supportthresholdnt}.
    \end{align*}
\end{claim}

Equipped with the above tool and the attack presented in \cref{sec:AttackingRobustProtocols}, we are ready to prove our main result.
\begin{proof} [Proof of \cref{thm:main_result}]
    Denote $t = \corruptionsrequired = \bigO{\sqrt{\numparties} \cdot \log^3 \numparties}$, consider the protocols $\Pi^0,\ldots,\Pi^t$ recursively defined by $ \Pi^0 \eqdef \Pi$ and $\Pi^{i+1} \eqdef \aprot{\Pi^i}{\Bc}$. If $\expecc{\Pi^t} < \epsnt$, then by \cref{lem:adaptive_adversary_composition} there exists a $t$-adaptive adversary that biases $\Pi$'s output to less than $\epsnt$ (the composition of all intermediate adversaries), and we are done. Else, by \cref{claim:arbitrarycoinflip:expec_oneshot} there exists $k \in [t]$ such that for $\Psi = \Pi^k$ it holds that
    \[
        \pr{\exists i \suchthat \minnn{\supp{\protincc{\Psi}{\Msggg{\le i}{\Psi}} \cond \Msggg{< i}{\Psi}}} \le -\supportthresholdnt} \le \delta.
    \]

    Hence, by \cref{thm:biasing_zero_robust_protocols}, there exists an $\amtcorruptions$-adaptive adversary \Ac
    such that
    \[
        \expecc{\aprot{\Psi}{\Ac}} \ge 1 - \epsnt.
    \]

    Denote by $\Cc$ the 
    $t$-adaptive adversary according to \cref{def:adaptive_adversary_composition} (the composition of all intermediate adversaries) such that $\Pi_{\Cc} \equiv \Pi^t$. Let $\Ac \circ \Cc$ be the $\bigO{\sqrt{\numparties} \cdot \log^3 \numparties}$-adaptive adversary according to \cref{def:adaptive_adversary_composition}, by \cref{lem:adaptive_adversary_composition} it holds that $\expecc{\aprot{\Pi}{\Ac \circ \Cc}} = \expecc{\aprot{\Psi}{\Ac}} \ge 1 - \epsnt$, concluding the proof.
\end{proof}

\section{Strongly Adaptive, Bidirectional Adversaries}\label{sec:StronglyAdaptive}
In this section, we use strongly adaptive adversaries to make the attacker described in the previous sections bidirectional (able to bias the protocol's outcome to both zero and one).
Formally (using the notations of \cref{sec:arbitrary_coinflip}), we prove the following result.
\restatementepslamdel

\begin{theorem} [Forcing full-information coin-flipping protocols] \label{thm:StronglyAdaptive}
    For any $\numparties$-party, full-information coin-flipping protocol $\Pi$ such that $\expecc{\Pi} \ge \epsnt$, there exists a $\bigO{\sqrt{\numparties} \cdot \log^3 \numparties}$-strongly adaptive adversary \Ac, such that $\expecc{\aprot{\Pi}{\Ac}} \ge 1 - \epsnt$.
\end{theorem}

Note that this is indeed a bidirectional attack capable of biasing protocols in which both values are significant enough. If one wants to bias a protocol towards 0, simply apply the attack on the flipped protocol (exactly the same protocol, besides the output function, which returns the opposite from the original output function).

\paragraph{The attack.}
The attack follows the same lines as the one described in \cref{sec:arbitrary_coinflip}, except in the immunization phase (\cref{alg:adversary_badjumps}) that turns the protocol to be robust. Using strongly adaptive corruptions, the adversary can always immunize the protocol so that the (non-strong) adaptive attack described in \cref{sec:AttackingRobustProtocols} is applicable. 
The strongly adaptive immunization deals with \emph{non-robust jumps} much better; instead of preparing for the worst-case scenario and corrupting every such jump, the attacker deals with them only if the unfavorable outcome is taken. 

\begin{algorithm}[One-shot strongly adaptive adversary $\Bc$ on $\Gamma$] \label{alg:strongadversary_badjumps}
    \item {\bf For} $i =1$ {\bf to} $\Length{\Gamma}$:
    \begin{enumerate}
        \item[] Let $\msgg{\le i}$ be the messages sent in the previous (and current) rounds, and let $\partition$ be the party that sent the \ith message. Let $\cM_i \eqdef \supp{\Msggg{\le i}{\Gamma} \cond \Msggg{< i}{\Gamma} = \msgg{< i}}$.
        
       \item[] {If} no message was corrupted before and $\protincc{\Gamma}{\msgg{\le i}} \le -\supportthresholdnt$, corrupt $\partition$ and instruct it send a message $m \in \cM_i$ with $\protincc{\Gamma}{\msgg{< i}, m} \ge 0$.
    \end{enumerate}
\end{algorithm}

\smallskip
\smallskip
Similarly to \cref{sec:arbitrary_coinflip}, consider the following (immediate) claim.
\begin{claim}\label{claim:properties_sadv0}
    Let $\Gamma$ be an $\numparties$-party, full-information coin-flipping protocol $\Gamma$. Then
    \begin{align*}
        \expecc{\aprot{\Gamma}{\Bc}} \ge \expecc{\Gamma} + \expecc{\mathrm{\#Corruptions}} \cdot \supportthresholdnt.
    \end{align*}
\end{claim}
\begin{proof}
	Immediate.
\end{proof}
Similarly to the immunization phase presented in \cref{sec:arbitrary_coinflip}, we iteratively apply $\Bc$ until
\begin{align}\label{eq:StronglyAdaptive:stop}
\pr{\exists i \suchthat \minnn{\supp{\protincc{\Gamma}{\Msggg{\le i}{\Gamma}} \cond \Msggg{< i}{\Gamma}}} \le -\supportthresholdnt} < \deltant / 2	
\end{align}
Denote the number of such applications by $t$. Let $\Pi^0 \eqdef \Pi$ and $\Pi^{i+1} \eqdef \aprot{\Pi^i}{\Bc}$. By \cref{claim:properties_sadv0}, it holds that $\expecc{\Pi^{i + 1}} \ge \expecc{\Pi^i} + \supportthresholdnt \cdot \expecc{\mathrm{\text{\#Corruptions in }} {\Pi^i}_{\Bc}}$. Reorganizing the terms, $\expecc{\mathrm{\text{\#Corruptions in }} {\Pi^{i}}_{\Bc}} \le (\expecc{\Pi^{i + 1}} - \expecc{\Pi^{i}}) \cdot \oneoversupportthresholdnt$. In total,
\begin{align}
	&\\
    &\sum_{i=0}^{t - 1} \expecc{\mathrm{\text{\#Corruptions in }} {\Pi^{i}}_{\Bc}} = \oneoversupportthresholdnt \cdot \sum_{i=0}^{t - 1} (\expecc{\Pi^{i}} - \expecc{\Pi^{i - 1}}) = \oneoversupportthresholdnt \cdot (\expecc{\Pi^t} - \expecc{\Pi^0}) \le \oneoversupportthresholdnt \nonumber
\end{align}

Let $\Cc$ be the composition of all intermediate adversaries, according to \cref{def:strongly_adaptive_adversary_composition}, such that $\Pi_{\Cc} \equiv \Pi^t$. By the above inequality, the expected amount of strongly adaptive corruptions $\Cc$ performs is at most $\oneoversupportthresholdnt$. Let $\Cc'$ be the variant of \Cc that aborts once it reaches $2\oneoversupportthresholdnt / \deltant$ corruptions, and denote $\Psi = \Pi_{\Cc'}$. By Markov's inequality, $\Cc'$ aborts with probability at most $\deltant / 2$. Combined with our stopping condition  (\ie \cref{eq:StronglyAdaptive:stop}) for applying \Bc, it follows that that 
\begin{align*}
\pr{\exists i \suchthat \minnn{\supp{\protincc{\Psi}{\Msggg{\le i}{\Psi}} \cond \Msggg{< i}{\Psi}}} \le -\supportthresholdnt} \le \deltant.
\end{align*}

In addition, $\expecc{\Psi} \ge \expecc{\Pi} \ge \epsnt$. Hence, by \cref{thm:biasing_zero_robust_protocols}, there exists an $\amtcorruptions$-strongly adaptive adversary \Ac such that \[\expecc{\aprot{\Psi}{\Ac}} \ge 1 - \epsnt.\]

Consider the attacker $\Ac \circ \Cc'$, the composed $\bigO{\sqrt{\numparties} \cdot \log^3 \numparties}$-strongly-adaptive adversary. By \cref{lem:strongly_adaptive_adversary_composition}, $\expecc{\aprot{\Pi}{\Ac \circ \Cc'}} = \expecc{\aprot{\Psi}{\Ac}} \ge 1 - \epsnt$, concluding the proof.
\subsection*{Acknowledgment}
We are grateful to Raz Landau, Nikolaos Makriyannis, Eran Omri, and Eliad Tsfadia for very helpful discussions. The first author is thankful to Michael Ben-Or for encouraging him to tackle this beautiful question.


\bibliographystyle{abbrvnat}
\bibliography{bib/crypto}

\end{document}